\documentclass{article}

\usepackage{amsmath}
\usepackage{amsthm}
\usepackage{amsfonts}

\usepackage{thmtools}
\usepackage{mathtools}

\usepackage{mathrsfs}

\usepackage{nameref}
\usepackage{hyperref}
\usepackage[capitalise]{cleveref}

\usepackage[inline]{enumitem}

\usepackage{relsize}
\usepackage[strict]{changepage}

\usepackage{microtype}

\usepackage{tikz}

\usetikzlibrary{arrows}
\usetikzlibrary{calc}
\usetikzlibrary{positioning}
\usetikzlibrary{shapes.geometric}

\tikzstyle{every picture}+=[remember picture,baseline]
\tikzstyle{arrow}+=[thick,rounded corners=0.5em]

\allowdisplaybreaks

\declaretheorem[
	name=Theorem,
	refname={theorem,theorems},
	Refname={Theorem,Theorems}
	]{theorem}
	
\declaretheorem[sibling=theorem,
	name=Definition,
	refname={definition,definitions},
	Refname={Definition,Definitions}
	]{definition}

\declaretheorem[
	unnumbered,
	name=Definition,
	refname={definition,definitions},
	Refname={Definition,Definitions}
	]{definition*}
	
\declaretheorem[sibling=theorem,
	name=Lemma,
	refname={lemma,lemmas},
	Refname={Lemma,Lemmas}
	]{lemma}

\declaretheorem[sibling=theorem,
	name=Corollary,
	refname={corollary,corollaries},
	Refname={Corollary,Corollaries}
	]{corollary}

\declaretheorem[
  unnumbered,
	name=Corollary,
	refname={corollary,corollaries},
	Refname={Corollary,Corollaries}
	]{corollary*}

\declaretheorem[
  unnumbered,
	name=Example,
	refname={example,examples},
	Refname={Example,Examples}
	]{example*}

\declaretheorem[
  unnumbered,
	name=Remark,
	refname={remark,remarks},
	Refname={Remark,Remarks}
	]{remark*}
	
\newlist{lemmaenum}{enumerate}{1}
\setlist[lemmaenum]{label=\arabic*., ref=\thelemma(\thelemmaenumi)}
\crefalias{lemmaenumi}{lemma} 

\renewcommand{\vec}[1]{\boldsymbol{#1}}

\newcommand{\defeq}{\mathbin{\stackrel{\mathclap{\normalfont\mbox{\smaller[3]{def}}}}{=}}}

\newcommand{\WLOG}{w.l.o.g.~}

\newcommand{\Nat}{\mathbb{N}}

\newcommand{\Ordinals}{\mathcal{O}}

\newcommand{\Antecedents}{\mathcal{A}}
\newcommand{\Consequents}{\mathcal{C}}
\newcommand{\Sequents}{\mathcal{S}}
\newcommand{\Models}{\mathcal{M}}
\newcommand{\TraceVals}{\mathcal{T}}

\newcommand{\Rule}{R}
\newcommand{\Rules}{\mathcal{R}}
\newcommand{\RuleIndices}{\mathbb{R}}

\newcommand{\node}{\nu}
\newcommand{\prf}{\mathcal{P}}
\newcommand{\traceval}{\tau}
\newcommand{\TracePairsOf}{\delta}

\newcommand{\CycleThreshold}{\mathbf{C}}
\newcommand{\TraceWidth}{\mathbf{W}}
\newcommand{\inDegree}{\mathbf{in}}

\newcommand{\isDefined}[1]{#1{\downarrow}}

\DeclareMathOperator{\NodesOf}{\mathsf{nodes}}
\DeclareMathOperator{\RuleOf}{\mathsf{rule}}
\DeclareMathOperator{\SeqOf}{\mathsf{seq}}
\DeclareMathOperator{\TraceValsOf}{\mathbb{T}}
\DeclareMathOperator{\OrdinalOf}{\Theta}
\DeclareMathOperator{\Proj}{proj}
\DeclareMathOperator{\ProgPointsOf}{prog}
\DeclareMathOperator{\maxStep}{\ensuremath{\TracePairsOf_{\mathsf{max}}}}

\DeclareMathOperator{\dom}{dom}

\begin{document}

\title{Size Relationships in \\ Abstract Cyclic Entailment Systems}
\author{
  Reuben N.~S.~Rowe \\[0.5em] 
  \textit{School of Computing} \\ 
  \textit{University of Kent, Canterbury, UK} \\
  \texttt{r.n.s.rowe@kent.ac.uk} \\[1em]
  James Brotherston \\[0.5em]
  \textit{Dept. of Computer Science} \\ 
  \textit{University College London, UK} \\
  \texttt{J.Brotherston@ucl.ac.uk}
  }
\date{February 2017}
\maketitle

\begin{abstract}
A \emph{cyclic} proof system generalises the standard notion of a proof as a finite tree of locally sound inferences by allowing proof objects to be potentially infinite. Regular infinite proofs can be finitely represented as \emph{graphs}. To preclude spurious cyclic reasoning, cyclic proof systems come equipped with a well-founded notion of `size' for the models that interpret their logical statements. A global soundness condition on proof objects, stated in terms of this notion of `size', ensures that any non-well-founded paths in the proof object can be disregarded.

We give an abstract definition of a subclass of such cyclic proof systems: cyclic \emph{entailment} systems. In this setting, we consider the problem of comparing the size of a model when interpreted in relation to the antecedent of an entailment, with that when interpreted in relation to the consequent. Specifically, we give a further condition on proof objects which ensures that models of a given entailment are always `smaller' when interpreted with respect to the consequent than when interpreted with respect to the antecedent. Knowledge of such relationships is useful in a program verification setting.
\end{abstract}

We consider the following abstract formulation of sequent-style proof systems for entailments between antecedents and consequents.

\begin{definition}[Abstract Entailment Proof Systems]
  Let $\Antecedents$ and $\Consequents$ be sets of \emph{antecedents} and \emph{consequents}, respectively. The set $\Sequents$ of \emph{sequents} is then the cartesian product $\Antecedents \times \Consequents$. We will write sequents $(A, C)$ as $A \vdash C$. A \emph{rule (schema)} $\Rule \subseteq \Sequents \times \Sequents^{\Nat}$ is a set of \emph{rule instances}. Rule instances $(S, \epsilon)$ are called \emph{axiomatic}. $\Rules = \wp(\Sequents \times \Sequents^{\Nat})$ is the set of all rules, which we may assume to be indexed by some set $\RuleIndices$. An \emph{entailment proof system} is a tuple $(\Sequents, \vec{r})$, where $\vec{r} \subseteq \RuleIndices$ identifies the rule schemas of the proof system.
\end{definition}

Here, we use bold-font vector notation for sequences, writing $\vec{s}_i$ for the $i$\textsuperscript{th} element of $\vec{s}$, $\epsilon$ for the empty sequence, $\vec{s_1} \cdot \vec{s_2}$ for sequence concatenation, and $|\vec{s}|$ for the number of elements in $\vec{s}$. We also write $\vec{s}_{i..j}$ to indicate that $\vec{s} = s_{i}, \ldots, s_{j}$ and we sometimes abuse notation by using $\vec{s}$ to refer to the set of elements occurring in $\vec{s}$.

A notion of \emph{cyclic} proof (as found, e.g., in \cite{Brotherston:05,Brotherston:07,Brotherston-Simpson:11}) can also be formulated in this abstract framework as follows.

\begin{definition}[Cyclic Pre-proofs]
  A (cyclic) pre-proof $\prf$ in an entailment proof system $(\Sequents, \vec{r})$ is a (finite) directed graph in which the children of each node $\node \in \NodesOf(\prf)$ are ordered. A node represents a rule instance via the functions $\SeqOf : \NodesOf(\prf) \rightarrow \Sequents$ and $\RuleOf : \NodesOf(\prf) \rightarrow \RuleIndices$, which satisfy $\RuleOf(\node) \in \vec{r}$ and $(\SeqOf(\node), (\SeqOf(\node_1), \ldots, \SeqOf(\node_n))) \in \Rules_{\RuleOf(\node)}$ for each node $\node \in \NodesOf(\prf)$ with children $\node_1, \ldots, \node_n$. In an abuse of notation, we may use $\node$ to refer to the rule instance it corresponds to. A \emph{path} in $\prf$ is a (possibly infinite) sequence of nodes $\vec{\node}$ of $\prf$ such that $\vec{\node}_{i}$ is the parent of $\vec{\node}_{i+1}$ for each non-terminal node $\vec{\node}_i$ in the path. We say that a path is \emph{rooted} at its first node.
\end{definition}

The semantics of such proof systems can be given by \emph{satisfaction relations} between elements of the proof system (i.e.~antecedents and consequents) and any appropriate notion of `model'.

\begin{definition}[Semantics]
  We fix a set $\Models$ of models and assume two satisfaction relations, ${\models_{\Antecedents}} \subseteq \Models \times \Antecedents$ and ${\models_{\Consequents}} \subseteq \Models \times \Consequents$, writing $m \models_{\Antecedents} A$ (resp.~$m \models_{\Consequents} C$) to mean $(m, A) \in {\models_{\Antecedents}}$ (resp.~$(m, C) \in {\models_{\Consequents}}$) for $m \in \Models$, $A \in \Antecedents$, and $C \in \Consequents$. We also write $m \models (A \vdash C)$ to mean that $m \models_{\Antecedents} A \Rightarrow m \models_{\Consequents} C$ holds. Thus, when we write $m \not\models (A \vdash C)$ we mean that $m \models_{\Antecedents} A$ but $m \not\models_{\Consequents} C$. We say that a sequent $A \vdash C$ is \emph{consistent} if there exists a model $m \models_{\Antecedents} A$, and \emph{inconsistent} otherwise.
\end{definition}

The definition of validity for sequents in our abstract framework is the standard one.

\begin{definition}[Validity]
  We say that a sequent $A \vdash C$ is \emph{valid} if and only if $m \models (A \vdash C)$ for all models $m \in \Models$. Alternatively, writing $\Models(A)$ and $\Models(C)$ for the set of models $m$ such that $m \models_{\Antecedents} A$ and $m \models_{\Consequents} C$, respectively, we may define $A \vdash C$ to be valid if and only if $\Models(A) \subseteq \Models(C)$.
\end{definition}

The minimum that we require of any proof system is that its proof rules are (locally) sound. The notion of local soundness has a uniform definition in terms of validity.

\begin{definition}[Local Soundness]
  We say that a rule instance $(S, (S_1, \ldots, S_n))$ is \emph{locally sound} if and only if whenever each $S_i$ is valid then $S$ is also valid. We say that a rule $\Rule$ is locally sound when all of its instances $(S, \vec{S}) \in \Rule$ are.
\end{definition}

Local soundness of the proof rules is not sufficient to ensure the soundness of \emph{cyclic} proof systems however. For this, we must also require that cyclic proofs satisfy a certain \emph{global} soundness property. In our general abstract framework, this can be defined using the following notion of \emph{traces}.

\begin{definition}[Trace Values]
  We fix disjoint sets $\TraceVals_{\Antecedents}$ and $\TraceVals_{\Consequents}$ of \emph{trace values}, and assume two functions $\TraceValsOf_{\Antecedents} : \Antecedents \rightarrow \wp_{\mathrm{fn}}(\TraceVals_{\Antecedents})$ and $\TraceValsOf_{\Consequents} : \Consequents \rightarrow \wp_{\mathrm{fn}}(\TraceVals_{\Consequents})$ that return a \emph{finite} set of trace values for each antecedent and consequent, respectively. In an abuse of notation, for a sequent $S = A \vdash C$, we may write $\TraceValsOf_{\Antecedents}(S)$ and $\TraceValsOf_{\Consequents}(S)$ to denote $\TraceValsOf_{\Antecedents}(A)$ and $\TraceValsOf_{\Consequents}(C)$ respectively. Then, in a further abuse of notation, for a node $\node$ in a cyclic pre-proof we write $\TraceValsOf_{\Antecedents}(\node)$ and $\TraceValsOf_{\Consequents}(\node)$ for $\TraceValsOf_{\Antecedents}(\SeqOf(\node))$ and $\TraceValsOf_{\Consequents}(\SeqOf(\node))$ respectively. Moreover, for a cyclic pre-proof $\prf$ we may write $\TraceValsOf_{\Antecedents}(\prf)$ to stand for the set $\bigcup_{\node \in \prf} \TraceValsOf_{\Antecedents}(\node)$, and write $\TraceValsOf_{\Consequents}(\prf)$ for the similarly defined set of consequent trace values appearing in $\prf$.
  We may drop the subscript and refer to both functions using $\TraceValsOf$ where the meaning is clear from the context.
\end{definition}

Let $\Ordinals$ denote (an initial segment of) the ordinals.

\begin{definition}[Trace Pair Functions]
  A \emph{trace pair function} $\vec{\TracePairsOf}$ is a (computable) family of pairs of functions, one for each (non-axiomatic) rule instance:
  \begin{multline*}
    \TracePairsOf^{(r, (S, \vec{S}_{1..n}), i)} : (\TraceValsOf_{\Antecedents}(S) \times \TraceValsOf_{\Antecedents}(\vec{S}_i) \rightharpoonup \Ordinals) \times (\TraceValsOf_{\Consequents}(S) \times \TraceValsOf_{\Consequents}(\vec{S}_i) \rightharpoonup \Ordinals)
  \end{multline*}
  where for each function $r \in \RuleIndices$, $(S, \vec{S}_{1..n}) \in \Rules_{r}$, and $i \in \{ 1, \ldots, n \}$. For convenience we will henceforth write $\TracePairsOf^{(r, (S, \vec{S}), i)}_{k}$ for $\Proj_{k}(\TracePairsOf^{(r, (S, \vec{S}), i)})$ where $k \in \{ 1, 2 \}$ and $\Proj_{k}$ projects the $k$\textsuperscript{th} element of a tuple, or simply $\TracePairsOf^{(r, (S, \vec{S}), i)}$ when the value of $k$ is clear from the context.
\end{definition}
  The intuition behind the trace pair functions is that the ordinal assigned to a trace pair represents a minimum distance between the `size' of any models that realize these trace values.
  
  When $(\traceval, \traceval')$ is in the domain of $\TracePairsOf^{(r, (S, \vec{S}), i)}_{k}$ (for $k \in \{ 1, 2 \}$), we say that $(\traceval, \traceval')$ is a (resp.~left-hand and right-hand) trace pair \emph{for} the rule instance $(S, \vec{S}) \in \Rules_{r}$ with respect to the $i$\textsuperscript{th} premise. If $\TracePairsOf^{(r, (S, \vec{S}), i)}_{k}(\traceval, \traceval') > 0$ then $(\traceval, \traceval')$ is called a \emph{progressing} trace pair, and it is called \emph{non-progressing} when $\TracePairsOf^{(r, (S, \vec{S}), i)}_{k}(\traceval, \traceval') = 0$. 
If $\node$ is a node in a cyclic pre-proof, with child nodes $\vec{\node'}_{\!\!1..m}$, then in an abuse of notation we may write $\TracePairsOf^{(\node, \vec{\node'}_{\!\!i})}_{k}$ for $\TracePairsOf^{(\RuleOf(\node), (\SeqOf(\node), (\SeqOf(\vec{\node'}_{\!\!1}), \ldots, \SeqOf(\vec{\node'}_{\!\!m}))), i)}_{k}$, and say that $(\traceval, \traceval')$ is a left-hand (resp.~right-hand) trace pair for $(\node, \vec{\node'}_{\!\!i})$ when $(\traceval, \traceval')$ is in the domain of $\TracePairsOf^{(\node, \vec{\node'}_{\!\!i})}_{k}$ for $k = 1$ (resp.~$k = 2$). For a left-hand (resp.~right-hand) trace value $\traceval \in \TraceValsOf(\node)$, we say that $\traceval$ is \emph{terminal} for $(\node, \vec{\node'}_{\!\!i})$ if there is no trace value $\traceval'$ such that $(\traceval, \traceval')$ is a left-hand (resp.~right-hand) trace pair for $(\node, \vec{\node'}_{\!\!i})$. We say that $\traceval$ is simply \emph{terminal for $\node$} when $\traceval$ is terminal for each $(\node, \vec{\node'}_{\!\!i})$.
  
  A \emph{cyclic} entailment proof system is a tuple $(\Sequents, \vec{r}, \TraceValsOf, \vec{\TracePairsOf})$. From now on, we shall assume some fixed cyclic entailment system.

\begin{definition}[Traces]
  A \emph{left-hand} (resp.~\emph{right-hand}) \emph{trace} $\vec{\traceval} \in \TraceVals_{\Antecedents}^{\omega}$ (resp.~$\vec{\traceval} \in \TraceVals_{\Consequents}^{\omega}$) is a (possibly infinite) sequence of trace values. We say that a left-hand (resp.~right-hand) trace \emph{follows} a path $\vec{\node}$ in a pre-proof $\prf$ when for every non-terminal value $\vec{\traceval}_{i}$ in the trace there are corresponding nodes $\vec{\node}_{i}$ and $\vec{\node}_{i+1}$ in the path such that $(\vec{\traceval}_{i}, \vec{\traceval}_{i+1})$ is a trace pair for $(\vec{\node}_{i}, \vec{\node}_{i+1})$. If $(\vec{\traceval}_{i}, \vec{\traceval}_{i+1})$ is a progressing trace pair for $(\vec{\node}_{i}, \vec{\node}_{i+1})$, then we say that the trace progresses at $i$. If the trace progresses at infinitely many points, then we say that it is infinitely progressing.
\end{definition}

\begin{definition}[Global Soundness]
\label{def:GlobalSoundness}
  A cyclic pre-proof $\prf$ is a valid cyclic \emph{proof} when every infinite path $\vec{\node} \in \prf$ has a tail that is followed by some infinitely progressing left-hand trace.
\end{definition}

This global soundness condition is decidable via a B\"{u}chi automata construction. If the proof system is sufficiently well-behaved -- specifically, if it admits an \emph{ordinal trace function} -- then validity of sequents occurring in cyclic proofs is guaranteed.

\begin{definition}[Ordinal Trace Function]
\label{def:OrdinalTraceFunction}
  An \emph{ordinal trace function} 
  is a partial function $\OrdinalOf : (\TraceVals_{\Antecedents} \cup \TraceVals_{\Consequents}) \times \Models \rightharpoonup \Ordinals$ which is (at least) defined on all $(\traceval, m)$ such that $\traceval \in \TraceVals_{\Antecedents}(A)$ and $m \models A$ for some antecedent $A$ and, dually, is undefined on all $(\traceval, m)$ such that $\traceval \in \TraceVals_{\Consequents}(C)$ .
  As usual, we write $\isDefined{\OrdinalOf(\traceval, m)}$ to denote that $(\traceval, m) \in \dom(\OrdinalOf)$.
\end{definition}

The ordinal trace function must also satisfy the following condition.

\begin{definition}[Descending Counter-model Property]
\label{def:DescendingCountermodels}
An ordinal trace function $\OrdinalOf$ satisfies the \emph{descending counter-model} property if and only if for all $r \in \vec{r}$:
  \begin{multline*}
    m \not\models S \wedge (S, \vec{S}) \in \Rules_{r} \Rightarrow {} \\
      \hspace{-5em} \exists \, m', {S_i \in \vec{S}} : m' \not\models S_i \wedge {} \\
        (\TracePairsOf^{(r, (S, \vec{S}), i)}_{1}(\traceval, \traceval') = \alpha \Rightarrow \OrdinalOf(\traceval', m') + \alpha \leq \OrdinalOf(\traceval, m))
  \end{multline*}
\end{definition}

The ordinal trace function can be seen as a \emph{realization} function that assigns realizers to trace values; thus models realize trace values. In fact, the ordinal trace function provides more information: it also tells us the `size' of each realization.
We note that the existence of an ordinal trace function 
entails local soundness, 
because of the requirement that falsifiability of the conclusion of a rule implies falsifiability of one of its premises. 

Let us assume that an ordinal trace function exists for our cyclic entailment system. The following soundness result holds.

\begin{theorem}[Soundness of Cyclic Proof Systems]
  Suppose $\prf$ is a valid cyclic proof of a sequent $S$ (i.e.~$\prf$ satisfies the global soundness condition of \Cref{def:GlobalSoundness}), then $S$ is valid.
\end{theorem}

We will now consider a further global condition on pre-proofs which allows us to relate antecedent trace values to consequent trace values. The intention is that, for valid cyclic proofs, this relation is \emph{sound} in the sense that 
the ordinal trace function respects this ordering. 

We begin by defining a \emph{maximality} property for right-hand traces. For this notion we fix a (decidable) predicate on rule instances $(A \vdash C, \vec{S})$ and right-hand trace values $\traceval \in \TraceVals(C)$, that we call the \emph{exclusion} predicate, with the property that $(A \vdash C, \vec{S})$ excludes $\traceval$ only if either there are no models $m$ such that $m\ \models A$ or there is no model $m$ of $A$ such that $\isDefined{\OrdinalOf(\traceval, m)}$.
For example, if the proof system contains an inconsistency axiom (i.e.~there exist no models of the sequent) then every instance of that axiom excludes all consequent trace values of the sequent.
We use this predicate to disregard right-hand traces which do not correspond to any model of the initial trace value.

\begin{definition}[Maximal Right-hand Traces]
\label{def:MaximalRighthandTraces}
  Let $\vec{\traceval}$ be a right-hand trace following a path $\vec{\node}$ in a (cyclic) pre-proof; we say that $\vec{\traceval}$ is \emph{maximal} when:
  \begin{enumerate*}[label=\roman*)]
    \item 
    it is finite (of length $n$, say); and
    \item 
    $\vec{\traceval}_{n}$ is terminal for $\vec{\node}_{n}$.
  \end{enumerate*}
  When $\vec{\node}_{n}$ excludes $\vec{\traceval}_{n}$ then we say $\vec{\traceval}$ is \emph{negative}; otherwise, we say it is \emph{positive}. In the case that $\vec{\node}_{n}$ is axiomatic, we say that $\vec{\traceval}$ is \emph{partially} maximal, and \emph{fully} maximal otherwise.
\end{definition}

We will need to ensure that maximal right-hand traces can be matched up with left-hand traces in a way that allows us to relate the sizes of the models that realize their initial trace values. To do so, we define two (semantic) notions for rules instances and sequents, respectively. We will write $\underline{1}$ for the least value taken by the ordinal trace function $\OrdinalOf$.

\begin{definition}[Grounded Traces]
  Let $(A \vdash C, \vec{S}) \in \Rules_{r}$ be an instance of rule $r$ and $\traceval \in \TraceValsOf(C)$ a trace value terminal for $(A \vdash C, \vec{S})$; we say that $\traceval$ is \emph{ground} with respect to $(A \vdash C, \vec{S})$ if whenever $A \vdash C$ is valid then $m \models_{\Antecedents} A$ implies that, when defined, $\OrdinalOf(\traceval, m) = \underline{1}$ for all models $m$. We say that a trace $\vec{\traceval}_{1..n}$ following a path $\vec{\node}$ is \emph{grounded} when its final trace value is ground (i.e.~when $\vec{\traceval}_{n}$ is ground with respect to $\vec{\node}_{n}$).
\end{definition}

Intuitively, this property allows us to determine when a maximal right-hand trace (following some path) ends in a trace value that is realized by a `smallest' model. Typically, the node in a path corresponding to the last value in such a maximal right-hand trace will be an instance of an unfolding rule. Some proof systems will allow maximal right-hand traces with final values whose realizers are not minimal, and we will want to ensure that we do not consider such traces. A typical example of such traces are those where the node in the proof corresponding to the final trace value is an instance of a right weakening rule.

We also define a family of relations between antecedent trace values $\TraceVals_{\Antecedents}$ and consequent trace values $\TraceVals_{\Consequents}$, indexed by sequents. The intuition is that related trace values have realizers of equal sizes.

\begin{definition}[Trace Values Equated by a Sequent]
  We say that a sequent $S \equiv A \vdash C$ \emph{equates} trace values $\traceval \in \TraceValsOf(A)$ and $\traceval' \in \TraceValsOf(C)$, and write $\traceval =_{S} \traceval'$, if whenever $S$ is valid then $m \models_{\Antecedents} A$ and $\isDefined{\OrdinalOf(\traceval', m)}$ implies that $\OrdinalOf(\traceval, m) = \OrdinalOf(\traceval', m)$ for all models $m$.
\end{definition}

The two properties that we have just defined will not, in general, be decidable for any given cyclic proof system. Therefore, in order to infer size relationships from cyclic proofs, we may need to approximate them. In the case of ground trace values, it will be sufficient if we can decide this for all instances of a given set of rules; then we can simply define no trace values to be ground with respect to instances of rules outside this set. Similarly, for trace values equated by a sequent, it will be sufficient for this to be decidable only for those sequents that are the conclusion of an \emph{axiomatic} rule instance.

To relate antecedent and consequent trace, we will be considering a \emph{quantitative} property of left- and right-hand traces, essentially counting the progression along each of the traces. We therefore define the following notion of size for traces.

\begin{definition}[Size of a Trace Along a Path]
\label{def:ProgressionSum}
  The \emph{size} $\ProgPointsOf_{\vec{\node}}(\vec{\traceval})$ of a \emph{finite} left-hand (resp.~right-hand) trace $\vec{\traceval}_{1..n}$ along a path $\vec{\node}$ which it follows is defined as follows, where $k = 1$ (resp.~$k = 2$):
  \begin{align*}
    \ProgPointsOf_{\vec{\node}}(\vec{\traceval}) & = 0 & (n = 1) \\
    \ProgPointsOf_{\vec{\node}}(\vec{\traceval}) & = \TracePairsOf^{(\vec{\node}_{n-1}, \vec{\node}_{n})}_{k}(\vec{\traceval}_{n-1}, \vec{\traceval}_{n}) + \ldots + \TracePairsOf^{(\vec{\node}_{1}, \vec{\node}_{2})}_{k}(\vec{\traceval}_{1}, \vec{\traceval}_{2}) & (n > 1)
  \end{align*}
\end{definition}
\noindent
Notice that this definition uses a \emph{reverse} sum. This is crucial for our result, as can been seen in the proof of \Cref{lem:TraceBoundLemma}, and is necessary because ordinal addition is \emph{not} commutative.

We can now define a trace value ordering relation with respect to pre-proofs. It is this relation that we intend to express the size relationship between the realizers (models) of trace values, and thus is at the heart of the result we present here.

\begin{definition}[Trace Value Ordering Relations]
\label{def:TraceValueOrdering}
  Let $\SeqOf(\node) = A \vdash C$ be a sequent in a cyclic proof $\prf$, with trace values $\traceval_1 \in \TraceValsOf(A)$ and $\traceval_2 \in \TraceValsOf(C)$; we will write $\traceval_2 \leq^{\node}_{\prf} \traceval_1$ whenever it holds that for all positive maximal right-hand traces $\vec{\traceval}_{1..n}$ with $\vec{\traceval}_1 = \traceval_2$ following paths $\vec{\node} \in \prf$ rooted at $\node$, there exists a left-hand trace $\vec{\traceval'}_{\!\!1..k}$, with $k \leq n$ and $\vec{\traceval'}_{\!\!1} = \traceval_1$, following $\vec{\node}$ such that:
  \begin{enumerate}[nosep,label=\roman*)]
    \item 
    $\ProgPointsOf_{\vec{\node}}(\vec{\traceval}) \leq \ProgPointsOf_{\vec{\node}}(\vec{\traceval'})$; and
    \item 
    either $\vec{\traceval}$ is grounded, or $\vec{\traceval}_n$ is partially maximal, $k = n$ and $\vec{\traceval'}_{\!\!n} =_{\SeqOf(\vec{\node}_n)} \vec{\traceval}_n$.
  \end{enumerate}
  If this condition holds with $\ProgPointsOf_{\vec{\node}}(\vec{\traceval}) < \ProgPointsOf_{\vec{\node}}(\vec{\traceval'})$, then we write $\traceval_2 <^{\node}_{\prf} \traceval_1$.
\end{definition}

We now proceed to define the extra property required of the ordinal trace function to be able to ensure that the trace value ordering relation is sound for valid proofs. In the following definition, we call a rule instance $(S, \vec{S})$ valid whenever the conclusion $S$ and each premise $\vec{S}_i$ are all valid sequents.

\begin{definition}[Descending Model Property]
  An ordinal trace function $\OrdinalOf$ satisfies the \emph{descending model} property if and only if for all valid, non-axiomatic rule instances $(A \vdash C, \vec{S}) \in \Rules_{r}$ ($r \in \vec{r}$), trace values $\traceval_1 \in \TraceValsOf(A)$ and $\traceval_2 \in \TraceValsOf(C)$, and models $m \models_{\Antecedents} A$, there exists a premise $A_i \vdash C_i \in \vec{S}$ and a model $m'$ such that:
  \begin{enumerate}[noitemsep,label={\roman*)}]
    \item 
    $m' \models_{\Antecedents} A_i$;
    \item 
    $\traceval_2$ is terminal for $(A \vdash C, \vec{S})$ or there is $\traceval'$ such that $\TracePairsOf^{(r, (A \vdash C, \vec{S}), i)}_{2}(\traceval_{2}, \traceval')$ is defined; and  
    \item 
    for all trace values $\traceval' \in \TraceVals$ and $k \in \{ 1, 2 \}$:
      \begin{multline*}
        \TracePairsOf^{(r, (A \vdash C, \vec{S}), i)}_{k}(\traceval_{k}, \traceval') = \alpha \wedge \isDefined{\OrdinalOf(\traceval_{k}, m)} \\
        {} \Rightarrow \isDefined{\OrdinalOf(\traceval', m')} \wedge \OrdinalOf(\traceval', m') + \alpha \sim_{k} \OrdinalOf(\traceval_{k}, m)
      \end{multline*}
    where $\sim_1 \defeq \leq$ and $\sim_2 \defeq \geq$
  \end{enumerate}
\end{definition}

Notice that this is different from the descending \emph{counter}-model property that we defined above, in that it says something about the \emph{models} of \emph{valid} sequents rather than \emph{counter}-models of \emph{in}valid ones, and that it also talks about consequent trace values. It asserts that the trace pair function soundly \emph{bounds} the difference in size between the realizations (i.e.~models) of trace pairs. In the case of antecedents this difference is bounded from above, and for consequents from below. As we shall see below, this means that the (reverse) sum of all the progression steps along a left-hand trace acts as a \emph{lower bound} on the size of realizations of the initial trace value; similarly, this sum for a right-hand trace serves as an \emph{upper} bound.

If an ordinal trace function satisfies the trace descent property above, then this is sufficient to guarantee that every model of a sequent in a valid cyclic proof corresponds to some positive maximal right-hand trace. We make this correspondence formal through the following notion of \emph{support}.

\begin{definition}[Trace Supports]
\label{def:TraceSupport}
  Let $\vec{\nu}$ be a path in a pre-proof $\prf$, where we have $\SeqOf(\vec{\nu}_i) \equiv A_i \vdash C_i$ for each node $\vec{\node}_i$ in the path, and let $\vec{\traceval}_{1..n}$ be a finite left-hand (resp.~right-hand) trace; we say that a sequence of models $\vec{m}_{1..k}$ ($k \geq n$) \emph{supports $\vec{\traceval}$ along $\vec{\nu}$} (and call $\vec{m}$ \emph{witness} of the support) if $\vec{\traceval}$ follows $\vec{\nu}$ and $\vec{m}_i \models_{\Antecedents} A_i$ (resp.~$\vec{m}_i \models_{\Consequents} C_i$) with $\isDefined{\OrdinalOf(\vec{\traceval}_{i}, \vec{m}_{i})}$ for each $0 < i \leq n$, and:
  \begin{gather*}
    \OrdinalOf(\vec{\traceval}_{j+1}, \vec{m}_{j+1}) + \TracePairsOf^{(\vec{\node}_{j}, \vec{\node}_{j+1})}_{k}(\vec{\traceval}_{j}, \vec{\traceval}_{j+1}) \sim \OrdinalOf(\vec{\traceval}_{j}, \vec{m}_{j})
  \end{gather*}
  for each $0 < j < n$ where $k = 1$ and ${\sim} \defeq {\leq}$ (resp.~$k = 2$ and ${\sim} \defeq {\geq}$). We may also say that a model $m$ supports a trace $\vec{\traceval}$ along a path $\vec{\node}$ when there exists a support witness $\vec{m}$ of $\vec{\traceval}$ along $\vec{\node}$ with $\vec{m}_1 = m$.
\end{definition}

We can show that if the ordinal trace function satisfies the trace descent property then every model of a sequent in a cyclic proof corresponds to (in the sense that it \emph{supports}) a (necessarily positive) maximal right-hand trace along some path rooted at that sequent. In fact, we have the stronger result that the same witness of this right-hand trace support additionally supports all left-hand traces that also follow this path.

\begin{lemma}[Supported Trace Existence]
\label{lem:SupportedTraceExistence}
  Assume the ordinal trace function satisfies the Trace Descent Property. If $A \vdash C \equiv \SeqOf(\node)$ is the sequent of a node $\node$ in a cyclic proof $\prf$ such that there is a trace value $\traceval' \in \TraceValsOf(C)$, and $m$ is a model such that $m \models_{\Antecedents} A$ and $\isDefined{\OrdinalOf(\traceval', m)}$, then there exists a path $\vec{\node}$ in $\prf$ rooted at $\node$, a positive maximal right-hand trace $\vec{\traceval'}_{\!\!1..n}$ beginning with $\traceval'$ and a witness $\vec{m}$ with $\vec{m}_1 = m$ such that $\vec{m}$ supports $\vec{\traceval'}$ along $\vec{\node}$; moreover $\vec{m}$ supports all left-hand traces $\vec{\traceval}$ that follow $\vec{\node}$.
\end{lemma}
\begin{proof}
  Since $\prf$ is a cyclic \emph{proof}, all the sequents it contains are valid. Using the Trace Descent Property, we can then show that there exists a path $\vec{\node} \in \prf$ rooted at $\node$, a right-hand trace $\vec{\traceval'}_{\!\!1..n}$ beginning with $\traceval'$ following $\vec{\node}$ and a witness $\vec{m}$ with $\vec{m}_1 = m$ such that 
  \begin{enumerate*}[label=\roman*)]
    \item 
    $\vec{m}$ supports $\vec{\traceval'}$ along $\vec{\node}$; and
    \item 
    $\vec{m}$ supports all left-hand traces $\vec{\traceval}$ that follow $\vec{\node}$.
  \end{enumerate*}
  Now, $\vec{\traceval'}$ cannot be infinite since if it were then, because $\prf$ is a valid cyclic proof, there would be an infinitely progressing left-hand trace following $\vec{\node}$ from which we could infer, by the properties of supported traces, the existence of an infinite descending chain of ordinals. Furthermore, there must exist a maximal such right-hand trace since if $\vec{\traceval'}_{\!\!1..n}$ were not maximal then one of the following two situations would hold. On the one hand, if $A_n \vdash C_n \equiv \SeqOf(\vec{\node}_n)$ is an instance of an axiom then it cannot be inconsistent since we have that $\vec{m}_n \models_{\Antecedents} A_n$; thus it must be that $\vec{\traceval'}_{n}$ is ground with respect to $\vec{\node}_{n}$, and so $\vec{\traceval'}$ is positve. On the other hand, if there exists a right-hand trace pair $(\vec{\traceval'}_{\!\!n}, \traceval'')$ for $\vec{\node}_n$ then by the Trace Descent Property, a support must exist for this longer trace.
\end{proof}

Supported traces have the following property, which expresses that the (reverse) sum of the progression steps in the trace bound (either from below, in the case of left-hand traces, or from above in the case of right-hand traces) the size of the models in the support as realizers of the values in the trace.

\begin{lemma}
\label{lem:TraceBoundLemma}
  If $\vec{m}$ supports a finite left-hand (resp.~right-hand) trace $\vec{\traceval}_{1..n}$ along a path $\vec{\node}$ then it follows that $\OrdinalOf(\vec{\traceval}_n, \vec{m}_n) + \ProgPointsOf_{\vec{\node}}(\vec{\traceval}) \sim \OrdinalOf(\vec{\traceval}_1, \vec{m}_1)$ where $\sim \defeq \leq$ (resp.~$\sim \defeq \geq$).
\end{lemma}
\begin{proof}
  By induction on the length of the trace $n$.
  \begin{description}[font=\normalfont]
    \item[($n = 1$):]
    Immediate, since then $\ProgPointsOf_{\vec{\node}}(\vec{\traceval}) = 0$, $\vec{\traceval}_1 = \vec{\traceval}_n$ and $\vec{\node}_1 = \vec{\node}_n$.
    \item[($n = k + 1$):] 
    Then $\vec{m}_{2..}$ supports $\vec{\traceval}_{2..n}$ along $\vec{\node}_{2..}$, so by the inductive hypothesis
    \begin{equation*}
      \OrdinalOf(\vec{\traceval}_n, \vec{m}_n) + \ProgPointsOf_{\vec{\node}_{2..}}(\vec{\traceval}_{2..n}) \sim \OrdinalOf(\vec{\traceval}_2, \vec{m}_2)
    \end{equation*}
    Therefore, since ordinal addition is monotone in the right argument
    \begin{multline*}
      \OrdinalOf(\vec{\traceval}_n, \vec{m}_n) + \ProgPointsOf_{\vec{\node}_{2..}}(\vec{\traceval}_{2..n}) + \TracePairsOf^{(\vec{\node}_1, \vec{\node}_2)}_{1}(\vec{\traceval}_1, \vec{\traceval}_2) \\ {} \sim \OrdinalOf(\vec{\traceval}_2, \vec{m}_2) + \TracePairsOf^{(\vec{\node}_1, \vec{\node}_2)}_{1}(\vec{\traceval}_1, \vec{\traceval}_2)
    \end{multline*}
    By \Cref{def:ProgressionSum}, this gives
    \begin{equation*}
      \OrdinalOf(\vec{\traceval}_n, \vec{m}_n) + \ProgPointsOf_{\vec{\node}}(\vec{\traceval}) \sim \OrdinalOf(\vec{\traceval}_2, \vec{m}_2) + \TracePairsOf^{(\vec{\node}_1, \vec{\node}_2)}_{1}(\vec{\traceval}_1, \vec{\traceval}_2)
    \end{equation*}
    Furthermore, by the definition of trace supports (\Cref{def:TraceSupport})
    \begin{equation*}
      \OrdinalOf(\vec{\traceval}_2, \vec{m}_2) + \TracePairsOf^{(\vec{\node}_1, \vec{\node}_2)}_{1}(\vec{\traceval}_1, \vec{\traceval}_2) \sim \OrdinalOf(\vec{\traceval}_1, \vec{m}_1)
    \end{equation*}
    Thus the result holds by the transitivity of $\leq$ for the ordinals.
    \qedhere
  \end{description}
\end{proof}

This leads to the soundness of the trace value ordering relations of a cyclic proof. Intuitively, this result holds because when two trace values are related by a trace value ordering relation, for any given model, the upper bound on its size as a realizer of the consequent trace value is \emph{not} greater (or is strictly smaller) than the \emph{lower} bound on its size as a realizer of the antecedent trace value.

\begin{theorem}[Soundness of Trace Value Ordering]
  Let $\SeqOf(\node) = A \vdash C$ be a sequent in a cyclic proof $\prf$, with trace values $\traceval_1 \in \TraceValsOf(A)$ and $\traceval_2 \in \TraceValsOf(C)$; then, for ${\sim} \in \{ <, \leq \}$
  \begin{gather*}
    \traceval_2 \sim^{\node}_{\prf} \traceval_1 \Rightarrow \forall \, m \in \Models : m \models_{\Antecedents} A \wedge \isDefined{\OrdinalOf(\traceval_2, m)} \Rightarrow \OrdinalOf(\traceval_2, m) \sim \OrdinalOf(\traceval_1, m)
  \end{gather*}
\end{theorem}
\begin{proof}
  Suppose that $\traceval_2 \sim^{\node}_{\prf} \traceval_1$ and take an arbitrary model such that $m \models_{\Antecedents} A$ and $\isDefined{\OrdinalOf(\traceval_2, m)}$. By \Cref{lem:SupportedTraceExistence}, there exists a path $\vec{\node} \in \prf$ rooted at $\node$, a positive maximal right-hand trace $\vec{\traceval'}_{\!\!1..n}$ beginning with $\traceval_2$ and a witness $\vec{m}$ with $\vec{m}_1 = m$ such that $\vec{m}$ supports $\vec{\traceval'}$ along $\vec{\node}$.
  \par
  Since $\traceval_2 \sim^{\node}_{\prf} \traceval_1$ there is also a left-hand trace $\vec{\traceval}_{1..k}$, with $k \leq n$ and $\vec{\traceval}_1 = \traceval_1$, following $\vec{\node}$ and satisfying $\ProgPointsOf_{\vec{\node}}(\vec{\traceval'}) \sim \ProgPointsOf_{\vec{\node}}(\vec{\traceval})$. By the trace support existence lemma, we also have that $\vec{m}$ supports $\vec{\traceval}$ along $\vec{\node}$.
  \par
  There are now two possibilities: $\vec{\traceval'}$ is either fully or partially maximal. In case of the former, $\vec{\traceval'}_{\!\!n}$ is ground with respect to $\vec{\node}_{n}$, so $\OrdinalOf(\vec{\traceval'}_{\!\!n}, \vec{m}_{n}) = \underline{1}$ and therefore $\OrdinalOf(\vec{\traceval'}_{\!\!n}, \vec{m}_{n}) \leq \OrdinalOf(\vec{\traceval}_{k}, \vec{m}_{k})$. In case of the latter, if $\vec{\traceval'}_{\!\!n}$ is not ground with respect to $\vec{\node}_{n}$, then $k = n$ and $\vec{\traceval}_{n} =_{\SeqOf(\vec{\node}_{n})} \vec{\traceval'}_{\!\!n}$ therefore $\OrdinalOf(\vec{\traceval}_{n}, \vec{m}_{n}) = \OrdinalOf(\vec{\traceval'}_{\!\!n}, \vec{m}_{n})$, and so also $\OrdinalOf(\vec{\traceval'}_{\!\!n}, \vec{m}_{n}) \leq \OrdinalOf(\vec{\traceval}_{k}, \vec{m}_{k})$. Thus, the result derives from the following chain of inequalities:
  \begin{align*}
    \OrdinalOf(\vec{\traceval'}_{\!\!1}, \vec{m}_{1}) & \leq \OrdinalOf(\vec{\traceval'}_{\!\!n}, \vec{m}_{n}) + \ProgPointsOf_{\vec{\node}}(\vec{\traceval'}) && \text{(\Cref{lem:TraceBoundLemma})} \\
    & \leq \OrdinalOf(\vec{\traceval}_{k}, \vec{m}_{k}) + \ProgPointsOf_{\vec{\node}}(\vec{\traceval'}) && \text{(weak left monot. of $+$)} \\
    & \sim \OrdinalOf(\vec{\traceval}_{k}, \vec{m}_{k}) + \ProgPointsOf_{\vec{\node}}(\vec{\traceval}) && \text{(right monot. of $+$)} \\
    & \leq \OrdinalOf(\vec{\traceval}_1, \vec{m}_{1}) && \text{(\Cref{lem:TraceBoundLemma})}
    \qedhere
  \end{align*}
\end{proof}

To decide whether $\traceval_2 \leq^{\node}_{\prf} \traceval_1$ holds for some given node in a cyclic proof, we will encode the problem as a language containment problem between \emph{weighted} automata. Although the language containment problem for weighted automata is known to be undecidable in general \cite{Krob94,AlmagorBK11}, it is decidable for the sub-class of finite-valued weighted sum-automata \cite{FiliotGR14}.

Weighted automata \cite{Droste2009} generalise finite-state automata by assigning to each word a value that is taken from a (usually infinite) set of \emph{weights}, rather than simply a binary value indicating whether the word is included in the language or not. A weighted automaton $\mathscr{A}$ over an alphabet $\Sigma$ and a semiring $(V, \oplus, \otimes)$ of \emph{weights} is a tuple $(Q, q_{I}, F, \Delta, \gamma)$ where $Q$ is a set of \emph{states}, $q_{I} \in Q$ is the \emph{initial} state, $F \subseteq Q$ is the set of \emph{final} states, $\Delta \subseteq Q \times \Sigma \times Q$ is the transition relation, and $\gamma : \Delta \rightarrow V$ is a function assigning a weight to each transition. A \emph{run} $\rho$ of $\mathscr{A}$ over a word $w = \sigma_1 \ldots \sigma_n \in \Sigma^{\ast}$ is a sequence $q_0 \sigma_1 q_1 \ldots \sigma_n q_n$ such that $q_0 = q_{I}$ and $(q_{i-1}, \sigma_{i}, q_{i}) \in \Delta$ for all $i \in \{ 1, \ldots, n \}$. We write $\rho : q_0 \xrightarrow{w} q_n$ to denote that $\rho$ is a run over $w$ starting at state $q_0$ and ending at $q_n$. A run $\rho : q_0 \xrightarrow{w} q_n$ is called \emph{accepting} if $q_n \in F$. 
The \emph{value} $\mathsf{V}(\rho)$ of a run $\rho : q_0 \xrightarrow{\sigma_1 \ldots \sigma_n} q_n$ is defined as the semiring product of the weight of each transition in the run, i.e.~$\mathsf{V}(\rho) = \gamma(q_0, \sigma_1, q_1) \otimes \ldots \otimes \gamma(q_{n-1}, \sigma_{n}, q_{n})$, if $\rho$ is accepting, and $\mathsf{V}(\rho) = \bot$ otherwise\footnote{Note that we can always augment a semiring $(V, \oplus, \otimes)$ with a fresh zero element $\bot$, which can be used to denote an `undefined' value.}. We write $R_{\mathscr{A}}(w)$ to denote the set $\{ \mathsf{V}(\rho) \mid \text{$\rho$ is a run of $\mathscr{A}$ on $w$} \}$. When $R_{\mathscr{A}}(w)$ is non-empty then we say that $w$ is \emph{in the domain of $\mathscr{A}$} and write $w \in \mathsf{dom}(\mathscr{A})$. The \emph{quantitative language} $L_{\mathscr{A}}$ defined by $\mathscr{A}$ is a function $L_{\mathscr{A}} : \Sigma^{\ast} \rightarrow V$, defined by $L_{\mathscr{A}}(w) = \bot$ if $w \not\in \mathsf{dom}(\mathscr{A})$, and $L_{\mathscr{A}}(w) = \bigoplus R_{\mathscr{A}}(w)$ otherwise. Given two weighted automata $\mathscr{A}$ and $\mathscr{B}$, we write $L_{\mathscr{A}} \leq L_{\mathscr{B}}$ if and only if $L_{\mathscr{A}}(w) \leq L_{\mathscr{B}}(w)$ for all words $w \in \Sigma^{\ast}$. Furthermore, we will write $L_{\mathscr{A}} < L_{\mathscr{B}}$ if and only if $L_{\mathscr{A}}(w) < L_{\mathscr{B}}(w)$ for all words $w$ such that $L_{\mathscr{A}}(w) \neq \bot$.

The semiring over which we will construct our automata is the ordinal-valued \emph{max-plus} (tropical) semiring $(\Ordinals_{\bot}, \oplus, \otimes)$ where $\alpha \oplus \beta = \max(\alpha, \beta)$ with $\bot < \alpha$ for all $\alpha \in \Ordinals$, and $\alpha \otimes \beta = \beta + \alpha$ (where $+$ is the usual addition on ordinals) with $\bot \otimes \alpha = \alpha \otimes \bot = \bot$ for all $\alpha \in \Ordinals$.

\begin{definition}
\label{def:AutomataConstruction:Full}
  Let $\prf$ be a cyclic proof, and define the alphabet $\Sigma_{\prf} = \NodesOf(\prf) \cup (\wp(\TraceValsOf_{\Antecedents}(\prf)) \times \TraceValsOf_{\Consequents}(\prf))$, and for $\traceval \in \TraceValsOf_{\Consequents}(\node)$ let $T_{\node}(\traceval)$ denote the set defined by $T_{\node}(\traceval) = \{ \traceval' \mid \traceval' \in \TraceValsOf_{\Antecedents}(\node) \wedge \traceval' =_{\SeqOf(\node)} \traceval \}$. For each node $\node_{\text{init}} \in \prf$ and trace values $\traceval_1 \in \TraceValsOf_{\Antecedents}(\node_{\text{init}})$ and $\traceval_2 \in \TraceValsOf_{\Consequents}(\node_{\text{init}})$, we can construct two weighted automata over $\Sigma_{\prf}$ and $(\Ordinals_\bot, \oplus, \otimes)$ as follows:
  \begin{align*}
    \mathscr{A}^{(\prf, \node_{\text{init}})}_{(\traceval_1, \traceval_2)} &= 
      (Q_{\Antecedents}, q_{\Antecedents}, F_{\Antecedents}, \Delta^{\Antecedents}_{\text{init}} \cup \Delta^{\Antecedents}_1 \cup \Delta^{\Antecedents}_2 \cup \Delta^{\Antecedents}_3 \cup \Delta^{\Antecedents}_4, \gamma_{\Antecedents}) \\
    \mathscr{B}^{(\prf, \node_{\text{init}})}_{(\traceval_1, \traceval_2)} &= 
      (Q_{\Consequents}, q_{\Consequents}, F_{\Consequents}, \Delta^{\Consequents}_{\text{init}} \cup \Delta^{\Consequents}_1 \cup \Delta^{\Consequents}_2, \gamma_{\Consequents})
  \end{align*}
  where
  \begin{align*}
      Q_{\Antecedents} &= (\NodesOf(\prf) \times \TraceValsOf_{\Antecedents}(\prf)) \uplus \{ \bot \} \uplus \{ \top \} \uplus \{ q_{\Antecedents} \}
      \\
      Q_{\Consequents} &= (\NodesOf(\prf) \times \TraceValsOf_{\Consequents}(\prf)) \uplus \{ \bot \} \uplus \{ q_{\Consequents} \}
      \\
      F_{\Antecedents} &= Q_{\Antecedents} \setminus \{ q_{\Antecedents} \}
      \\
      F_{\Consequents} &= 
        \begin{aligned}[t]
          \{ \bot \} & \cup \{ (\node, \traceval) \mid \text{$\node$ not axiomatic, $\traceval$ terminal for $\node$,} \\
            & \hspace{5.1cm} \text{and $\node$ does not exclude $\traceval$} \} \\
           & \cup \{ (\node, \traceval) \mid \text{$\node$ axiomatic, $\traceval$ ground with respect to $\node$,} \\
            & \hspace{5.1cm} \text{and $\node$ does not exclude $\traceval$} \}
        \end{aligned}
      \\[0.5em]
      \Delta^{\Antecedents}_{\text{init}} &= \{ (q_{\Antecedents}, \node_{\text{init}}, (\node_{\text{init}}, \traceval_1)) \} \\
      \Delta^{\Antecedents}_1 &= \{ ((\node, \traceval), \node', (\node', \traceval')) \mid (\node, \node') \in \prf \wedge (\traceval, \traceval') \in \dom(\TracePairsOf^{(\node, \node')}_{1}) \}
      \\
      \Delta^{\Antecedents}_2 &= \{ ((\node, \traceval), (T_{\node}(\traceval'), \traceval'), \bot) \mid \text{$\node$ axiomatic} \wedge \traceval' \in \TraceValsOf_{\Consequents}(\node) \wedge \traceval \in T_{\node}(\traceval') \}
      \\
      \Delta^{\Antecedents}_3 &= \{ ((\node, \traceval), \node', \top) \mid (\node, \node') \in \prf \}
      \\
      \Delta^{\Antecedents}_4 &= \{ (\top, \node, \top) \mid \node \in \NodesOf(\prf) \}
      \\[0.5em]
      \Delta^{\Consequents}_{\text{init}} &= \{ (q_{\Consequents}, \node_{\text{init}}, (\node_{\text{init}}, \traceval_2)) \} \\
      \Delta^{\Consequents}_1 &= \{ ((\node, \traceval), \node', (\node', \traceval')) \mid (\node, \node') \in \prf \wedge (\traceval, \traceval') \in \dom(\TracePairsOf^{(\node, \node')}_{2}) \}
      \\
      \Delta^{\Consequents}_2 &= \{ ((\node, \traceval), (T_{\node}(\traceval), \traceval), \bot) \mid \text{$\traceval \in \TraceValsOf_{\Consequents}(\node)$, $\node$ axiomatic,} \\
        & \hspace{3.25cm} \text{$\traceval$ not ground w.r.t. $\node$, $\node$ does not exclude $\traceval$} \, \}
  \end{align*}
  \begin{gather*}
    \gamma_{\Antecedents}(q, \sigma, q') = \left\{ \begin{aligned}
        & \TracePairsOf^{(\node, \node')}_{1}(\traceval, \traceval') && \text{if $q = (\node, \traceval)$ and $q = (\node', \traceval')$} \\
        & 0 && \text{otherwise}
      \end{aligned} \right.
    \\[0.5em]
    \gamma_{\Consequents}(q, \sigma, q') = \left\{ \begin{aligned}
        & \TracePairsOf^{(\node, \node')}_{2}(\traceval, \traceval') && \text{if $q = (\node, \traceval)$ and $q = (\node', \traceval')$} \\
        & 0 && \text{otherwise}
      \end{aligned} \right.
  \end{gather*}
  We say that the automaton $\mathscr{B}^{(\prf, \node)}_{(\traceval_1, \traceval_2)}$ is \emph{grounded} whenever it holds that $\traceval$ is ground with respect to each \emph{reachable} final state $(\node, \traceval) \in F_{\Consequents}$.
\end{definition}

The language of $\mathscr{B}^{(\prf, \node)}_{(\traceval_1, \traceval_2)}$ contains all and only the paths in $\prf$ rooted at $\node$ followed by positive maximal right-hand traces starting with $\traceval_2$, and the value of a word is the maximum of the sizes of all such right-hand traces following the path it corresponds to. Furthermore, the language of $\mathscr{A}^{(\prf, \node)}_{(\traceval_1, \traceval_2)}$ contains all sequences of nodes (in $\prf$) containing a prefix which is a path in $\prf$ rooted at $\node$ and followed by a (maximal) left-hand trace starting with $\traceval_1$. In particular, it contains \emph{all} (finite) paths in $\prf$ rooted at $\node$, and the value of a path (word) is the maximum of the sizes of all left-hand traces following it. Thus, provided that $\mathscr{B}^{(\prf, \node)}_{(\traceval_1, \traceval_2)}$ is grounded, this construction then ensures that $L_{\mathscr{B}^{(\prf, \node)}_{(\traceval_1, \traceval_2)}} \leq L_{\mathscr{A}^{(\prf, \node)}_{(\traceval_1, \traceval_2)}}$ if and only if $\traceval_2 \leq_{\prf}^{\node} \traceval_1$.

\begin{lemma}
\label{lem:Automaton:Antecedent:Runs}
  If $\vec{\node}_{1..n}$ is a path in a cyclic proof $\prf$ and $\rho : q_{\Antecedents} \xrightarrow{\vec{\node}} q_{n}$ is a run of $\mathscr{A}^{(\prf, \node)}_{(\traceval, \traceval')}$ then there exists a left-hand trace $\vec{\traceval}_{1..k}$ ($k \leq n$) with $\vec{\traceval}_{1} = \traceval$ such that $q_i = (\vec{\node}_{i}, \vec{\traceval}_{i})$ for each $i \in \{ 1, \ldots, k \}$ and $q_{j} = \top$ for each $j \in \{ k+1, \ldots, n \}$. Moreover, if $q_n = (\vec{\node}_{n}, \traceval'')$ for some trace value $\traceval''$ then $k = n$ and $\traceval'' = \vec{\traceval}_{n}$.
\end{lemma}
\begin{proof}
  By induction on $n$.
\end{proof}

\begin{lemma}
\label{lem:Automaton:Faithful:Antecedent}
  Let $\vec{\node}_{1..n}$ be a path in a cyclic proof $\prf$ and let $\vec{\traceval}_{1..k}$ ($k \leq n$) be a left-hand trace; then $\vec{\traceval}$ follows $\vec{\node}$ if and only if there exists a run $\rho : q_{\Antecedents} \xrightarrow{\vec{\node}} q_{n}$ of $\mathscr{A}^{(\prf, \vec{\node}_{1})}_{(\vec{\traceval}_{1}, \traceval')}$ such that $q_{i} = (\vec{\node}_{i}, \vec{\traceval}_{i})$ for each $i \in \{ 1, \ldots, k \}$. Moreover it holds that $\mathsf{V}(\rho) \geq \ProgPointsOf_{\vec{\node}}(\vec{\traceval})$ and, furthermore, if $k = n$ or $q_{k+1} = \top$ then $\ProgPointsOf_{\vec{\node}}(\vec{\traceval}) \geq \mathsf{V}(\rho)$.
\end{lemma}
\begin{proof}
  By induction on $n$.
\end{proof}

\begin{lemma}
\label{lem:Automaton:Faithful:Consequent}
  Let $\vec{\node}$ be a finite sequence of nodes in a cyclic proof $\prf$ and let $\vec{\traceval}_{1..n}$ be a right-hand trace; then $\vec{\node}$ is a path in $\prf$ followed by $\vec{\traceval}$ if and only if $\rho = q_{\Consequents}\,\vec{\node}_{1}\,(\vec{\node}_{1}, \vec{\traceval}_{1})\,\ldots\,\vec{\node}_{n}\,(\vec{\node}_{n}, \vec{\traceval}_{n})$ is a run of $\mathscr{B}^{(\prf, \vec{\node}_{1})}_{(\traceval', \vec{\traceval}_{1})}$. Moreover $\mathsf{V}(\rho) = \ProgPointsOf_{\vec{\node}}(\vec{\traceval})$.
\end{lemma}
\begin{proof}
  By induction on $n$.
\end{proof}

\begin{corollary}
\label{cor:Automaton:Faithful:Consequent}
  $\vec{\traceval}_{1..n}$ is a positive maximal right-hand trace following $\vec{\node}$ in $\prf$ if and only if $\rho$ is an accepting run of $\mathscr{B}^{(\prf, \vec{\node}_{1})}_{(\traceval', \vec{\traceval}_{1})}$
such that either:
  \begin{enumerate}[nosep,label=\roman*)]
    \item 
    $\rho = q_{\Consequents}\,\vec{\node}_{1}\,(\vec{\node}_{1}, \vec{\traceval}_{1})\,\ldots\,\vec{\node}_{n}\,(\vec{\node}_{n}, \vec{\traceval}_{n})$ (in which case either $\vec{\traceval}$ is fully maximal or $\vec{\traceval}_{n}$ is ground with respect to $\vec{\node}_{n}$); or 
    \item
    $\rho = q_{\Consequents}\,\vec{\node}_{1}\,(\vec{\node}_{1}, \vec{\traceval}_{1})\,\ldots\,\vec{\node}_{n}\,(\vec{\node}_{n}, \vec{\traceval}_{n})\,(T_{\vec{\node}_n}(\vec{\traceval}_{n}), \vec{\traceval}_{n})\,\bot$ (in which case $\vec{\traceval}$ is partially maximal). 
  \end{enumerate}
  Moreover it holds that $\mathsf{V}(\rho) = \ProgPointsOf_{\vec{\node}}(\vec{\traceval})$.
\end{corollary}
\begin{proof}
  Immediately from \cref{lem:Automaton:Faithful:Consequent} and the definitions of maximal right-hand traces (\cref{def:MaximalRighthandTraces}) and the automata constructions (\cref{def:AutomataConstruction:Full}).
\end{proof}

\begin{theorem}[Soundness and Completeness of the Automata Construction]
\label{thm:TraceValueOrdering-Automata-Equivalence}
  Let $\prf$ be a cyclic proof, $\node$ be a node in $\prf$, and $\traceval_{\Antecedents} \in \TraceValsOf_{\Antecedents}(\node)$ and $\traceval_{\Consequents} \in \TraceValsOf_{\Consequents}(\node)$ be trace values; then, for ${\sim} \in \{<, \leq \}$, $\traceval_{\Consequents} \sim^{\node}_{\prf} \traceval_{\Antecedents}$ if and only if $L_{\mathscr{B}^{(\prf, \node)}_{(\traceval_{\Antecedents}, \traceval_{\Consequents})}} \sim L_{\mathscr{A}^{(\prf, \node)}_{(\traceval_{\Antecedents}, \traceval_{\Consequents})}}$ and $\mathscr{B}^{(\prf, \node)}_{(\traceval_{\Antecedents}, \traceval_{\Consequents})}$ is grounded.
\end{theorem}
\begin{proof}
  \begin{description}
    \item[(if):]
    Let $\vec{\traceval}_{1..n}$ be a positive maximal right-hand trace with $\vec{\traceval}_1 = \traceval_{\Consequents}$ following a path $\vec{\node}$ in $\prf$ rooted at $\node$. 
    We now consider the following two (exhaustive) possibilities:
    \begin{description}
      \item[($\vec{\traceval}$ is partially maximal):]
      thus $\vec{\node}_{n}$ is axiomatic and by \cref{cor:Automaton:Faithful:Consequent}, there is an accepting run $\rho$ of $\mathscr{B}^{(\prf, \node)}_{(\traceval_{\Antecedents}, \traceval_{\Consequents})}$ such that $\mathsf{V}(\rho) = \ProgPointsOf_{\vec{\node}}(\vec{\traceval})$ and one of the following two cases holds:
      \begin{enumerate}[label=\roman*),ref=\roman*]
        \item 
        \label{enumitem:AutomataConstructionSoundness:case:RighthandTraceRooted}
        $\rho = q_{\Consequents}\,\vec{\node}_{1}\,(\vec{\node}_{1}, \vec{\traceval}_{1})\,\ldots\,\vec{\node}_{n}\,(\vec{\node}_{n}, \vec{\traceval}_{n})$ and so $\vec{\traceval}_{n}$ is ground with respect to $\vec{\node}_{n}$. Since $L_{\mathscr{B}^{(\prf, \node)}_{(\traceval_{\Antecedents}, \traceval_{\Consequents})}} \sim L_{\mathscr{A}^{(\prf, \node)}_{(\traceval_{\Antecedents}, \traceval_{\Consequents})}}$, there exists an accepting run $\rho' : q_{\Antecedents} \xrightarrow{\vec{\node}_{1..n}} q$ of $\mathscr{A}^{(\prf, \node)}_{(\traceval_{\Antecedents}, \traceval_{\Consequents})}$ such that $\mathsf{V}(\rho) \sim \mathsf{V}(\rho')$. By \cref{lem:Automaton:Antecedent:Runs,lem:Automaton:Faithful:Antecedent} it then follows that there exists a left-hand trace $\vec{\traceval'}_{1..k}$ ($k \leq n$) with $\vec{\traceval'}_{\!\!1} = \traceval_{\Antecedents}$ following $\vec{\node}$ such that $\mathsf{V}(\rho') = \ProgPointsOf_{\vec{\node}}(\vec{\traceval'})$, and thus that $\ProgPointsOf_{\vec{\node}}(\vec{\traceval}) \sim \ProgPointsOf_{\vec{\node}}(\vec{\traceval'})$. Therefore, the conditions of \cref{def:TraceValueOrdering} are met for $\vec{\traceval}$ and $\vec{\node}$.
        \item 
        $\rho = q_{\Consequents}\,\vec{\node}_{1}\,(\vec{\node}_{1}, \vec{\traceval}_{1})\,\ldots\,\vec{\node}_{n}\,(\vec{\node}_{n}, \vec{\traceval}_{n})\,(T_{\vec{\node}_n}(\vec{\traceval}_{n}), \vec{\traceval}_{n})\,\bot$. Then, since $L_{\mathscr{B}^{(\prf, \node)}_{(\traceval_{\Antecedents}, \traceval_{\Consequents})}} \sim L_{\mathscr{A}^{(\prf, \node)}_{(\traceval_{\Antecedents}, \traceval_{\Consequents})}}$, it follows $\rho' : q_{\Antecedents} \xrightarrow{\vec{\node}_{1..n} \cdot (T_{\vec{\node}_n}(\vec{\traceval}_{n}), \vec{\traceval}_{n})} q$ is an accepting run of $\mathscr{A}^{(\prf, \node)}_{(\traceval_{\Antecedents}, \traceval_{\Consequents})}$ such that $\mathsf{V}(\rho) \sim \mathsf{V}(\rho')$. By construction, $q = \bot$ and $\rho'' : q_{\Antecedents} \xrightarrow{\vec{\node}_{1..n}} (\vec{\node}_{n}, \traceval')$ is a run of $\mathscr{A}^{(\prf, \node)}_{(\traceval_{\Antecedents}, \traceval_{\Consequents})}$ for some $\traceval'$ and therefore, by \cref{lem:Automaton:Antecedent:Runs,lem:Automaton:Faithful:Antecedent}, there exists a left-hand trace $\vec{\traceval'}_{\!\!1..n}$ with $\vec{\traceval'}_{\!\!1} = \traceval_{\Antecedents}$ and $\vec{\traceval'}_{\!\!n} = \traceval'$ such that $\vec{\traceval'}$ follows $\vec{\node}$ and $\mathsf{V}(\rho'') = \ProgPointsOf_{\vec{\node}}(\vec{\traceval'})$. Also by construction, we have that $\traceval' \in T_{\vec{\node}_{n}}(\vec{\traceval}_{n})$ and therefore that $\vec{\traceval'}_{\!\!n} =_{\SeqOf(\vec{\node}_{n})} \vec{\traceval}_{n}$. Now notice
        \begin{align*}
          \mathsf{V}(\rho') & = \mathsf{V}(\rho'') \otimes \gamma_{\Antecedents}((\vec{\node}_{n}, \vec{\traceval'}_{\!\!n}), (T_{\vec{\node}_n}(\vec{\traceval}_{n}), \vec{\traceval}_{n}), \bot) \\
          & = \mathsf{V}(\rho'') \otimes 0 \\
          & = 0 + \mathsf{V}(\rho'') = \mathsf{V}(\rho'')
        \end{align*}
      and therefore that $\ProgPointsOf_{\vec{\node}}(\vec{\traceval}) \sim \ProgPointsOf_{\vec{\node}}(\vec{\traceval'})$. Thus, the conditions of \cref{def:TraceValueOrdering} are met for $\vec{\traceval}$ and $\vec{\node}$.
      \end{enumerate}
      \item[($\vec{\traceval}$ is fully maximal):]
      thus $\vec{\node}_{n}$ is axiomatic and by \cref{cor:Automaton:Faithful:Consequent}, there is an accepting run $\rho = q_{\Consequents}\,\vec{\node}_{1}\,(\vec{\node}_{1}, \vec{\traceval}_{1})\,\ldots\,\vec{\node}_{n}\,(\vec{\node}_{n}, \vec{\traceval}_{n})$ of $\mathscr{B}^{(\prf, \node)}_{(\traceval_{\Antecedents}, \traceval_{\Consequents})}$ such that $\mathsf{V}(\rho) = \ProgPointsOf_{\vec{\node}}(\vec{\traceval})$. Since $\mathscr{B}^{(\prf, \node)}_{(\traceval_{\Antecedents}, \traceval_{\Consequents})}$ is grounded it follows that $\vec{\traceval}_{n}$ is ground with respect to $\vec{\node}_{n}$. The remainder of this case proceeds as for item (\ref*{enumitem:AutomataConstructionSoundness:case:RighthandTraceRooted}) above.
    \end{description}
    Thus, since the conditions of \cref{def:TraceValueOrdering} are met for an arbitrary positive maximal right-hand trace $\vec{\traceval}$ and path $\vec{\node}$ in $\prf$, it follows that $\traceval_{\Consequents} \sim^{\node}_{\prf} \traceval_{\Antecedents}$.
    \item[(only if):]
    We first show $L_{\mathscr{B}^{(\prf, \node)}_{(\traceval_{\Antecedents}, \traceval_{\Consequents})}} \sim L_{\mathscr{A}^{(\prf, \node)}_{(\traceval_{\Antecedents}, \traceval_{\Consequents})}}$. Take any $w \in \dom\left(\mathscr{B}^{(\prf, \node)}_{(\traceval_{\Antecedents}, \traceval_{\Consequents})}\right)$; there must exist a maximally valued run $\rho$ of $\mathscr{B}^{(\prf, \node)}_{(\traceval_{\Antecedents}, \traceval_{\Consequents})}$ over $w$. Then $L_{\mathscr{B}^{(\prf, \node)}_{(\traceval_{\Antecedents}, \traceval_{\Consequents})}}(w) = \mathsf{V}(\rho)$ since $\rho$ is maximally valued. Notice that $\rho$ may take one of two forms:
    \begin{description}
      \item[$\rho = q_{\Consequents} \, \node_1 \, (\node_1, \traceval_1) \, \ldots \, \node_n \, (\node_n, \traceval_n)$] with $\node_1 = \node$ and $\traceval_1 = \traceval_{\Consequents}$. Then by \cref{cor:Automaton:Faithful:Consequent} $\vec{\node} = \node_1 \, \ldots \, \node_n$ is a path in $\prf$ followed by right-hand trace $\vec{\traceval} = \traceval_1 \, \ldots \, \traceval_n$ with $\vec{\traceval}$ a positive maximal trace, and $\mathsf{V}(\rho) = \ProgPointsOf_{\vec{\node}}(\vec{\traceval})$. Since $\traceval_{\Consequents} \sim^{\node}_{\prf} \traceval_{\Antecedents}$ it follows that there exists a left-hand trace $\vec{\traceval'}_{\!\!1..k}$ following $\vec{\node}$ with $k \leq n$ and $\vec{\traceval'}_{\!\!1} = \traceval_{\Antecedents}$, and furthermore that $\ProgPointsOf_{\vec{\node}}(\vec{\traceval}) \sim \ProgPointsOf_{\vec{\node}}(\vec{\traceval'})$. Therefore by \cref{lem:Automaton:Faithful:Antecedent} there is a run $\rho' : q_{\Antecedents} \xrightarrow{\vec{\node}} q_n$ of $\mathscr{A}^{(\prf, \node)}_{(\traceval_{\Antecedents}, \traceval_{\Consequents})}$ such that $\ProgPointsOf_{\vec{\node}}(\vec{\traceval'}) \leq \mathsf{V}(\rho')$. Notice that $\rho'$ is an accepting run since all (non-initial) states in $\mathscr{A}^{(\prf, \node)}_{(\traceval_{\Antecedents}, \traceval_{\Consequents})}$ are accepting. So, by definition, $\mathsf{V}(\rho') \leq L_{\mathscr{A}^{(\prf, \node)}_{(\traceval_{\Antecedents}, \traceval_{\Consequents})}}(\vec{\node})$. Therefore
      \begin{multline*}
        L_{\mathscr{B}^{(\prf, \node)}_{(\traceval_{\Antecedents}, \traceval_{\Consequents})}}(\vec{\node}) = \mathsf{V}(\rho) = \ProgPointsOf_{\vec{\node}}(\vec{\traceval}) \\
        {} \sim \ProgPointsOf_{\vec{\node}}(\vec{\traceval'}) \leq \mathsf{V}(\rho') \leq L_{\mathscr{A}^{(\prf, \node)}_{(\traceval_{\Antecedents}, \traceval_{\Consequents})}}(\vec{\node})
      \end{multline*}
      
      \item[$\rho = q_{\Consequents} \, \node_1 \, (\node_1, \traceval_1) \, \ldots \, \node_n \, (\node_n, \traceval_n) \, (T_{\vec{\node}_n}(\vec{\traceval}_{n}), \vec{\traceval}_{n}) \, \bot$] with $\node_1 = \node$ and $\traceval_1 = \traceval_{\Consequents}$. Then by \cref{cor:Automaton:Faithful:Consequent} $\vec{\node} = \node_1 \, \ldots \, \node_n$ is a path in $\prf$ followed by right-hand trace $\vec{\traceval} = \traceval_1 \, \ldots \, \traceval_n$ with $\vec{\traceval}$ a positive maximal trace and $\mathsf{V}(\rho) = \ProgPointsOf_{\vec{\node}}(\vec{\traceval})$. By construction, since there is a transition from $(\node_n, \traceval_n)$ to $\bot$, it follows that $\node_n$ is axiomatic and that $\traceval_n$ is \emph{not} ground with respect to $\node_n$. Thus, since $\traceval_{\Consequents} \sim^{\node}_{\prf} \traceval_{\Antecedents}$, there exists a left-hand trace $\vec{\traceval'}_{\!\!1..n}$ following $\vec{\node}$ with $\vec{\traceval'}_{\!\!1} = \traceval_{\Antecedents}$ such that $\ProgPointsOf_{\vec{\node}}(\vec{\traceval}) \sim \ProgPointsOf_{\vec{\node}}(\vec{\traceval'})$ and $\vec{\traceval'}_{\!\!n} =_{\SeqOf(\vec{\node}_{n})} \vec{\traceval}_{n}$; it therefore follows that $\vec{\traceval'}_{\!\!n} \in T_{\vec{\node}_{n}}(\vec{\traceval}_{n})$. Now by \cref{lem:Automaton:Faithful:Antecedent} there exists a run $\rho' = q_{\Antecedents} \, \vec{\node}_{1} \, (\vec{\node}_{1}, \vec{\traceval'}_{\!\!1}) \, \ldots \, \vec{\node}_{n} \, (\vec{\node}_{n}, \vec{\traceval'}_{\!\!n})$ such that $\mathsf{V}(\rho') = \ProgPointsOf_{\vec{\node}}(\vec{\traceval'})$. Moreover by construction there is a transition in $\mathscr{A}^{(\prf, \node)}_{(\traceval_{\Antecedents}, \traceval_{\Consequents})}$ from $(\vec{\node}_{n}, \vec{\traceval'}_{\!\!n})$ to $\bot$ via $(T_{\vec{\node}_n}(\vec{\traceval}_{n}), \vec{\traceval}_{n})$. Thus $\rho''$ is a run of $\mathscr{A}^{(\prf, \node)}_{(\traceval_{\Antecedents}, \traceval_{\Consequents})}$ over $w$, where 
      \begin{equation*}
        \rho'' = q_{\Antecedents} \, \vec{\node}_{1} \, (\vec{\node}_{1}, \vec{\traceval'}_{\!\!1}) \, \ldots \, \vec{\node}_{n} \, (\vec{\node}_{n}, \vec{\traceval'}_{\!\!n}) \, (T_{\vec{\node}_n}(\vec{\traceval}_{n}), \vec{\traceval}_{n}) \, \bot
      \end{equation*}        
      Notice that it is an accepting run since all (non-initial) states in $\mathscr{A}^{(\prf, \node)}_{(\traceval_{\Antecedents}, \traceval_{\Consequents})}$ are accepting, and moreover that by construction $\mathsf{V}(\rho'') = \mathsf{V}(\rho')$. Also, by definition it holds that $\mathsf{V}(\rho'') \leq L_{\mathscr{A}^{(\prf, \node)}_{(\traceval_{\Antecedents}, \traceval_{\Consequents})}}(w)$. Therefore
      \begin{multline*}
        L_{\mathscr{B}^{(\prf, \node)}_{(\traceval_{\Antecedents}, \traceval_{\Consequents})}}(w) = \mathsf{V}(\rho) = \ProgPointsOf_{\vec{\node}}(\vec{\traceval}) \\
        {} \sim \ProgPointsOf_{\vec{\node}}(\vec{\traceval'}) = \mathsf{V}(\rho'') \leq L_{\mathscr{A}^{(\prf, \node)}_{(\traceval_{\Antecedents}, \traceval_{\Consequents})}}(w)
      \end{multline*}
      
    \end{description}
    To see that $\mathscr{B}^{(\prf, \node)}_{(\traceval_{\Antecedents}, \traceval_{\Consequents})}$ is grounded, consider an arbitrary reachable final state $(\node', \traceval')$ of $\mathscr{B}^{(\prf, \node)}_{(\traceval_{\Antecedents}, \traceval_{\Consequents})}$ where $\node'$ is not axiomatic (in the case that $\node'$ is axiomatic we have by construction that $\traceval'$ is ground with respect to $\node'$). Since $(\node', \traceval')$ is reachable, there is accepting run $q = q_{\Consequents} \, \node_1 \, (\node_1, \traceval_1) \, \ldots \, \node_n \, (\node_n, \traceval_n)$ with $\node' = \node_n$ and $\traceval' = \traceval_n$. Thus by \cref{cor:Automaton:Faithful:Consequent} it follows that $\vec{\node} = \node_1 \, \ldots \, \node_n$ is a path in $\prf$ followed by positive maximal right-hand trace $\vec{\traceval} = \traceval_1 \, \ldots \, \traceval_n$. Since $\traceval_{\Consequents} \sim^{\node}_{\prf} \traceval_{\Antecedents}$ it follows that $\traceval_n$ is ground with respect to $\node_n$. Since the choice of reachable final state was arbitrary, we have that $\mathscr{B}^{(\prf, \node)}_{(\traceval_{\Antecedents}, \traceval_{\Consequents})}$ is grounded.
    \qedhere
  \end{description}
\end{proof}

We now show that when certain constraints are placed on the trace pair function $\vec{\TracePairsOf}$, the constructions of \cref{def:AutomataConstruction:Full} result in automata that are \emph{finitely ambiguous}. That is, the number of runs of the automata on any given word is bounded. Finitely ambiguous automata (over certain semirings) have decidable containment \cite{FiliotGR14}.

\begin{definition}
  We say that a trace pair function $\vec{\TracePairsOf}$ is \emph{trace injective} if every $\TracePairsOf^{(r, (S, \vec{S}), i)}_{k}$ satisfies the following, for every $\traceval$, $\traceval'$, and $\traceval''$:
  \begin{equation*}
    (\traceval', \traceval) \in \dom(\TracePairsOf^{(r, (S, \vec{S}), i)}_{k}) \wedge (\traceval'', \traceval) \in \dom(\TracePairsOf^{(r, (S, \vec{S}), i)}_{k}) \Rightarrow \traceval' = \traceval''
  \end{equation*}
\end{definition}

\begin{lemma}
\label{lem:TraceInjectivityImpliesBoundedTraces}
  Suppose $\vec{\TracePairsOf}$ is trace injective and let $\node$ be a node in a cyclic proof $\prf$ with $\traceval \in \TraceValsOf_{\Antecedents}(\node)$ (resp.~$\traceval \in \TraceValsOf_{\Consequents}(\node)$) a trace value in $\node$; then for all paths $\vec{\node}_{1..n}$ rooted at $\node$ and left-hand (resp.~right-hand) traces $\vec{\traceval}_{1..n}$ and $\vec{\traceval'}_{1..n}$ following $\vec{\node}$ with $\vec{\traceval}_1 = \vec{\traceval'}_1 = \traceval$, if $\vec{\traceval}_n = \vec{\traceval'}_n$ then $\vec{\traceval} = \vec{\traceval'}$.
\end{lemma}
\begin{proof}
  By induction on $n$. The base case ($n = 1$) is immediate. For the inductive step, assume a path $\vec{\node}_{1..(i+1)}$ rooted at $\node$ followed by left-hand (resp.~right-hand) traces $\vec{\traceval}_{1..(i+1)}$ and $\vec{\traceval'}_{1..(i+1)}$ with $\vec{\traceval}_1 = \vec{\traceval'}_1 = \traceval$; furthermore, assume that $\vec{\traceval}_{i+1} = \vec{\traceval'}_{i+1}$. Since both $\vec{\traceval}$ and $\vec{\traceval'}$ follow $\vec{\node}$ we have that $(\vec{\traceval}_{i}, \vec{\traceval}_{i+1}) \in \dom(\TracePairsOf^{(\vec{\node}_{i}, \vec{\node}_{i+1})}_{k})$ and $(\vec{\traceval'}_{i}, \vec{\traceval'}_{i+1}) \in \dom(\TracePairsOf^{(\vec{\node}_{i}, \vec{\node}_{i+1})}_{k})$ for $k = 1$ (resp.~$k = 2$). Then, since $\vec{\traceval}_{i+1} = \vec{\traceval'}_{i+1}$ and $\TracePairsOf$ is trace injective, it follows that $\vec{\traceval}_{i} = \vec{\traceval'}_{i}$. Then by the inductive hypothesis we have that $\vec{\traceval}_{1..i} = \vec{\traceval'}_{1..i}$, whence the result follows.
\end{proof}

The significance of \cref{lem:TraceInjectivityImpliesBoundedTraces} is that when the trace pair function is trace injective there is a bound on the number of possible traces following any given path. As an immediate corollary, every automaton $\mathscr{B}^{(\prf, \node)}_{(\traceval_1, \traceval_2)}$ is finitely ambiguous. 
However, even when the trace pair function is trace injective it is \emph{not} necessarily the case that a left-hand trace automaton $\mathscr{A}^{(\prf, \node)}_{(\traceval_1, \traceval_2)}$ is finitely ambiguous. 
The reason for this is the presence of the `sink' state $\top$.
When a proof contains a (left-hand) trace cycle (of the form $(n_1, \traceval_1) \ldots (n_{j}, \traceval_{j})$ with nodes $n_{1} = n_{j}$ and trace values $\traceval_{1} = \traceval_{j}$), the resulting left-hand trace automaton will contain the following configuration of states:
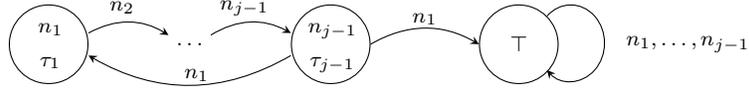
\begin{center}
  \footnotesize
  \begin{tikzpicture}
    \node[draw,circle,inner sep=0pt,minimum size=3.5em] (firststate) at (-5cm,0) {\makebox{$\begin{gathered}n_{1}\\\traceval_{1}\end{gathered}$}};
    \node (dots) at (-3.125cm,0) {$\ldots$} ;
    \node[draw,circle,inner sep=0pt,minimum size=3.5em] (finalstate) at (-1.25cm,0) {\makebox{$\begin{gathered}n_{j-1}\\\traceval_{j-1}\end{gathered}$}};
    \node[draw,circle,inner sep=0pt,minimum size=3.5em] (topstate) at (1.25cm,0) {\makebox{$\top$}};
    \node (topstatetransitions) at (3.5cm,0) {$n_1, \ldots, n_{j-1}$} ;
    \path[draw,>=stealth,->] ($(firststate.east)+(0,0.5em)$) [bend left] node [above, xshift=1.5em, yshift=0.5em] {$n_2$} to (dots.north west) ;
    \path[draw,>=stealth,->] (dots.north east) [bend left] node [above, xshift=1.5em, yshift=0.5em] {$n_{j-1}$} to ($(finalstate.west)+(0,0.5em)$) ;
    \path[draw,>=stealth,->] ($(finalstate.west)+(0,-0.5em)$) [bend left] node[below, xshift=-1.25cm, yshift=-0.25em] {$n_1$} to ($(firststate.east)+(0,-0.5em)$) ;
    \path[draw,>=stealth,->] (finalstate.east) [bend left] node [midway, above, yshift=0.5em] {$n_1$} to (topstate.west) ;
    \path[draw,>=stealth,->] (topstate.north east) arc[x radius=-1.5em, y radius=-1.7em, start angle=-45, end angle=-315] ;
  \end{tikzpicture}
\end{center}
That is, there are runs $(n_{j-1}, \traceval_{j-1}) \xrightarrow{w} (n_{j-1}, \traceval_{j-1})$, $(n_{j-1}, \traceval_{j-1}) \xrightarrow{w} \top$, and $\top \xrightarrow{w} \top$ with $w = n_1 \ldots n_{j-1}$. This results in the automaton being infinitely ambiguous \cite[{\textsection}3]{WeberSeidl91} and thus when the weight of the cycle is non-zero it is also infinite-valued.

Nonetheless, it is possible to define a sequence of finitely ambiguous weighted automata $\mathscr{A}(n)^{(\prf, \node)}_{(\traceval_1, \traceval_2)}$ (where $0 < n \in \Nat$) that constitute successively better approximations to $\mathscr{A}^{(\prf, \node)}_{(\traceval_1, \traceval_2)}$. For proofs $\prf$ satisfying certain restrictions, we can show that there exists $n$ such that we can use $\mathscr{A}(n)^{(\prf, \node)}_{(\traceval_1, \traceval_2)}$ to decide the trace value ordering relation $\sim^{\node}_{\prf}$.

\begin{definition}
\label{def:AutomataConstruction:Approximate}
  Let $\prf$ be a cyclic proof and take the alphabet $\Sigma_{\prf}$ and sets $T_{\node}(\traceval)$ as in \cref{def:AutomataConstruction:Full}. Then, for each node $\node_{\text{init}} \in \prf$, trace values $\traceval_1 \in \TraceValsOf_{\Antecedents}(\node_{\text{init}})$ and $\traceval_2 \in \TraceValsOf_{\Consequents}(\node_{\text{init}})$, and every $n > 0$, define the weighted automaton $\mathscr{A}(n)^{(\prf, \node_{\text{init}})}_{(\traceval_1, \traceval_2)}$ as follows:
  \begin{gather*}
    \mathscr{A}(n)^{(\prf, \node_{\text{init}})}_{(\traceval_1, \traceval_2)}
      = (Q, q_{\Antecedents}, F, \Delta_{\text{init}} \cup \Delta_{1} \cup \Delta_{2} \cup \Delta_{3} \cup \Delta_{4} \cup \Delta_{5}, \gamma)
  \end{gather*}
  where
  \vspace{-2em}
  \begin{align*}
    \phantom{\text{where }}
    Q &= 
      \begin{aligned}[t]
        (\NodesOf(\prf) \times \TraceValsOf_{\Antecedents}(\prf)) &\uplus \{ \bot \} \uplus \{ q_{\Antecedents} \} 
        \\ 
        &\uplus \{ \top_{\hspace{-0.25em}\node}^{i} \,\mid\, \node \in \NodesOf(\prf), 0 < i \leq n \} 
      \end{aligned}
    \\
    F &= Q \setminus \{ q_{\Antecedents} \} 
    \\
    \Delta_{\text{init}} &= \{ (q_{\Antecedents}, \node_{\text{init}}, (\node_{\text{init}}, \traceval_1)) \} 
    \\
    \Delta_{1} &= \{ ((\node, \traceval), \node', (\node', \traceval')) \mid (\node, \node') \in \prf \wedge (\traceval, \traceval') \in \dom(\TracePairsOf^{(\node, \node')}_{1}) \}
    \\
    \Delta_{2} &= \{ ((\node, \traceval), (T_{\node}(\traceval'), \traceval'), \bot) \mid \text{$\node$ axiomatic} \wedge \traceval' \in \TraceValsOf_{\Consequents}(\node) \wedge \traceval \in T_{\node}(\traceval') \}
    \\
    \Delta_{3} &= \{ ((\node, \traceval), \node', \top_{\hspace{-0.25em}\node'}^{1}) \mid (\node, \node') \in \prf \}
    \\
    \Delta_{4} &= \{ (\top_{\hspace{-0.25em}\node}^{i}, \node', \top_{\hspace{-0.25em}\node}^{i}) \,\mid\, 0 < i \leq n, \node' \in \NodesOf(\prf), \node' \neq \node \} 
    \\
    \Delta_{5} &= \{ (\top_{\hspace{-0.25em}\node}^{i}, \node, \top_{\hspace{-0.25em}\node}^{i+1}) \,\mid\, 0 < i < n \}
  \end{align*}
  \begin{gather*}
    \gamma(q, \sigma, q') = \left\{ \begin{aligned}
        & \TracePairsOf^{(\node, \node')}_{1}(\traceval, \traceval') && \text{if $q = (\node, \traceval)$ and $q = (\node', \traceval')$} \\
        & 0 && \text{otherwise}
      \end{aligned} \right.
  \end{gather*}
\end{definition}
Note that these `approximate' automata all share a common kernel with the `full' automaton: the only difference is in the number of `sink' states that they admit. The construction of \Cref{def:AutomataConstruction:Approximate} refines the sink state of the full construction into a collection of finite chains of sink states. The crucial difference is that each chain \emph{remembers} which letter of the alphabet (i.e.~node in the proof) was encountered on entry. The chain then serves to only allow words containing a finite number of successive occurrences of that node (i.e.~once a left-hand trace terminates, only paths which visit the next node of the proof up to $n$ times, and not more, are accepted). It is this which results in the approximate automata being finitely ambiguous. The full automaton, in contrast, permits an unbounded number of successive occurrences. Thus, in an abuse of notation, we will also sometimes find it convenient to write $\mathscr{A}(\omega)^{(\prf, \node)}_{(\traceval, \traceval')}$ for $\mathscr{A}^{(\prf, \node)}_{(\traceval, \traceval')}$.

\begin{lemma}
\label{lem:ApproximateAutomata:FiniteAmbiguity}
  When the trace pair function $\vec{\TracePairsOf}$ is trace injective, each approximate automaton $\mathscr{A}(n)^{(\prf, \node)}_{(\traceval_1, \traceval_2)}$ is finitely ambiguous.
\end{lemma}
\begin{proof}
  As for the automata $\mathscr{B}^{(\prf, \node)}_{(\traceval_1, \traceval_2)}$ this follows from \cref{lem:TraceInjectivityImpliesBoundedTraces}, which gives that there is a unique run ending in a state of the form $(\node, \traceval)$ for any given word.
  We use $\TraceWidth(\prf)$ to denote the maximum number of trace values occurring in the antecedent or consequent of any node in $\prf$, and refer to this as the \emph{trace width} of $\prf$. We also use $\inDegree(\prf)$ to denote the graph-theoretic notion of the \emph{in-degree} of $\prf$ (i.e.~the maximum number of predecessors for any given node in $\prf$).  With \cref{lem:TraceInjectivityImpliesBoundedTraces}, it is easy to see that the maximum number of runs of $\mathscr{A}(n)^{(\prf, \node)}_{(\traceval_1, \traceval_2)}$ that end in the state $\bot$ for any given word is $|\NodesOf(\prf)| \times \TraceWidth(\prf)$. Similarly making use of \cref{lem:TraceInjectivityImpliesBoundedTraces}, it can be shown by a straightforward induction on the number of states in the run of the form $\top^{n}_{\node}$ that maximum number of runs ending in such a state for any given word has a maximum bound of $\inDegree(\prf) \times \TraceWidth(\prf)$.
  (The bound is realised immediately in the base case; the inductive case is a trivial application of the inductive hypothesis since there can only be a single, unique transition to the final state for any given letter of the alphabet).
\end{proof}

The approximate automata are sound

\begin{lemma}
\label{lem:ApproximateAutomata:Soundness}
  $\rho : q_{\Antecedents} \xrightarrow{\vec{\sigma}_{1..m}} q_{m}$ is a run of $\mathscr{A}(n)^{(\prf, \node)}_{(\traceval, \traceval')}$ only if $\rho' : q_{\Antecedents} \xrightarrow{\vec{\sigma}_{1..m}} q'_{m}$ is a run of $\mathscr{A}(\omega)^{(\prf, \node)}_{(\traceval, \traceval')}$, with $q'_{i} = \top$ if $q_{i} \equiv \top^{j}_{\node}$ (for some $j$ and $\node$), and $q'_{i} = q_{i}$ otherwise (for all $0 < i \leq m$); moreover, $\mathsf{V}(\rho') = \mathsf{V}(\rho)$. 
\end{lemma}
\begin{proof}
  By induction on the length $m$ of the run.
\end{proof}

\begin{corollary}
\label{cor:ApproximateAutomata:Soundness}
  $L_{\mathscr{B}^{(\prf, \node)}_{(\traceval, \traceval')}} \sim L_{\mathscr{A}(n)^{(\prf, \node)}_{(\traceval, \traceval')}}$ implies that $L_{\mathscr{B}^{(\prf, \node)}_{(\traceval, \traceval')}} \sim L_{\mathscr{A}(\omega)^{(\prf, \node)}_{(\traceval, \traceval')}}$, for ${\sim} \in \{ {<}, {\leq} \}$.
\end{corollary}
\begin{proof}
  Suppose $\rho : q_0 \xrightarrow{\vec{\sigma}} q$ is a run of $\mathscr{B}^{(\prf, \node)}_{(\traceval, \traceval')}$ such that $\mathsf{V}(\rho)$ is a positive maximal trace. Since $L_{\mathscr{B}^{(\prf, \node)}_{(\traceval, \traceval')}} \sim L_{\mathscr{A}(n)^{(\prf, \node)}_{(\traceval, \traceval')}}$, there is a run $\rho' : q'_0 \xrightarrow{\vec{\sigma}} q'$ of $\mathscr{A}(n)^{(\prf, \node)}_{(\traceval, \traceval')}$ such that $\mathsf{V}(\rho) \sim \mathsf{V}(\rho')$. Then, by \cref{lem:ApproximateAutomata:Soundness}, there is a run $\rho'' : q''_0 \xrightarrow{\vec{\sigma}} q''$ of $\mathscr{A}(\omega)^{(\prf, \node)}_{(\traceval, \traceval')}$ such that $\mathsf{V}(\rho') = \mathsf{V}(\rho'')$, whence the result follows.
\end{proof}

%

We also prove a relative completeness lemma which will be useful later in showing full completeness for the restricted set of proofs that we will consider.

\begin{lemma}
\label{lem:ApproximateAutomata:RelativeCompleteness}
  Let $\rho : q_{\Antecedents} \xrightarrow{\vec{\sigma}_{1..m}} q_{m}$ be a run of $\mathscr{A}(\omega)^{(\prf, \node)}_{(\traceval, \traceval')}$ and k be the number of occurrences of $\vec{\sigma}_{i+1}$ in the sequence $\vec{\sigma}_{(i+1)..m}$, where $i$ is the least such that $q_{i+1} \equiv \top$ (we take $k = 0$ when there exists no such $i$); then $\rho' : q_{\Antecedents} \xrightarrow{\vec{\sigma}_{1..m}} q'_{m}$ is a run of $\mathscr{A}(n)^{(\prf, \node)}_{(\traceval, \traceval')}$ for each $n \geq k$, with $q'_{j} = q_{j}$ if $j \leq i$ (or $j \leq m$ if no such $i$ exists) and $q'_{j} = \top^{k'}_{\vec{\sigma}_{i+1}}$ otherwise, where $k'$ is the number of occurrences of $\vec{\sigma}_{i+1}$ in the sequence $\vec{\sigma}_{(i+1)..j}$, and moreover $\mathsf{V}(\rho) = \mathsf{V}(\rho')$.
\end{lemma}
\begin{proof}
  By induction on the length $m$ of the run.
\end{proof}

To define our restrictions for decidability we will make use of standard graph-theoretic notions of reachability and simple cycles, defined concretely for our formalisation of cyclic proofs as follows. If $\node$ and $\node'$ are nodes in a cyclic proof $\prf$, and $\traceval \in \TraceValsOf(\node)$ and $\traceval' \in \TraceValsOf(\node')$ are left-hand (resp.~right-hand) trace values, then we say that $(\node', \traceval')$ is \emph{reachable} from $(\node, \traceval)$ when there exists a path $\vec{\node}_{1..n} \in \prf$ with $\vec{\node}_{1} = \node$ and $\vec{\node}_{n} = \node'$ that is followed by some left-hand (resp.~right-hand) trace $\vec{\traceval}_{1..n}$ with $\vec{\traceval}_{1} = \traceval$ and $\vec{\traceval}_{n} = \traceval'$. We say that $\vec{\traceval}_{1..n}$ is a left (resp.~right) cycle along a path $\vec{\node}_{1..n}$ ($n > 1$) in $\prf$ to mean that $\vec{\traceval}$ is a left-hand (resp.~right-hand) trace following $\vec{\node}$ such that $\vec{\node}_{1} = \vec{\node}_{n}$ and $\vec{\traceval}_{1} = \vec{\traceval}_{n}$; we say that the cycle is rooted at $\vec{\node}_{1}$. Furthermore, we say that $\vec{\traceval}$ is a \emph{simple} cycle when, for all $1 < i < j \leq n$, it is not the case that $\vec{\node}_{i} = \vec{\node}_{j}$ and $\vec{\traceval}_{i} = \vec{\traceval}_{j}$. Notice that these definitions exactly correspond with reachability and (simple) cycles in the automata constructed from $\prf$. We extend slightly the notion of cycle and say that the pair $(\vec{\traceval}_{1..n}, \vec{\traceval'}_{1..n})$ is a \emph{binary} cycle along a path $\vec{\node}_{1..n}$ when both $\vec{\traceval}$ and $\vec{\traceval'}$ are cycles along $\vec{\node}$. Similarly, $(\vec{\traceval}, \vec{\traceval'})$ is a simple binary cycle when, for all $1 < i < j \leq n$, it is not the case that $\vec{\node}_{i} = \vec{\node}_{j}$, with $\vec{\traceval}_{i} = \vec{\traceval}_{j}$ and $\vec{\traceval'}_{i} = \vec{\traceval'}_{j}$. In this case notice that, individually, neither $\vec{\traceval}$ nor $\vec{\traceval'}$ need necessarily be simple cycles along $\vec{\node}$.

We now come to describe the restrictions that we place on proofs in order to obtain decidability of trace value ordering relations. 

\begin{definition}
\label{def:CompletenessRestrictions}
  Let $\prf$ be a cyclic proof; we say that, with respect to some given node $\node_{\text{init}}$ in $\prf$ and trace values $\traceval_{\Antecedents} \in \TraceValsOf_{\Antecedents}(\node_{\text{init}})$ and $\traceval_{\Consequents} \in \TraceValsOf_{\Consequents}(\node_{\text{init}})$, $\prf$ is:
  \begin{enumerate}[label={\roman*)},ref={\roman*}]
    \item 
    \label{def:CompletenessRestrictions:FinitelyProg}
    \emph{left finitely progressing} when for each edge $(\node, \node') \in \prf$ and pair of left-hand trace values $\traceval \in \TraceValsOf_{\Antecedents}(\node)$ and $\traceval' \in \TraceValsOf_{\Antecedents}(\node')$ such that $(\node, \traceval)$ is reachable from $(\node_{\text{init}}, \traceval_{\Antecedents})$ and $(\node, \node') \in \dom(\TracePairsOf^{(\node, \node')})$, it is the case that $\TracePairsOf^{(\node, \node')}(\traceval, \traceval')$ is a finite ordinal, i.e.~$\TracePairsOf^{(\node, \node')}(\traceval, \traceval') < \omega$. We define $\prf$ to be \emph{right finitely progressing} analogously, and say that $\prf$ is simply \emph{finitely progressing} when it is both left and right finitely progressing.
    \item 
    \label{def:CompletenessRestrictions:Dynamic}
    \emph{dynamic} when 
    it is the case that $\ProgPointsOf_{\vec{\node}}(\vec{\traceval}) > 0$ for every left and right simple cycle $\vec{\traceval}$ 
    along $\vec{\node}$ such that $(\vec{\node}_{1}, \vec{\traceval}_{1})$ is reachable from $(\node_{\text{init}}, \traceval_{\Antecedents})$ and $(\node_{\text{init}}, \traceval_{\Consequents})$ respectively.
    \item 
    \label{def:CompletenessRestrictions:Balanced}
    \emph{balanced} when $\ProgPointsOf_{\vec{\node}}(\vec{\traceval}) = \ProgPointsOf_{\vec{\node}}(\vec{\traceval'})$ for every simple left binary cycle $(\vec{\traceval}_{1..n}, \vec{\traceval'}_{1..n})$ along $\vec{\node}_{1..n}$ such that $(\vec{\node}_{1}, \vec{\traceval}_{1})$ and $(\vec{\node}_{1}, \vec{\traceval'}_{1})$ are both reachable from $(\node_{\text{init}}, \vec{\traceval}_{\Antecedents})$.
  \end{enumerate}
\end{definition}

The first condition ensures that the automata we construct from $\prf$ are \emph{sum} automata. The second and third conditions, along with the first, ensure completeness of the restricted automata for the trace value ordering relations. 
Note that all three conditions are decidable, since there are only finitely many edges and simple (binary) cycles in any given proof.

We will show that when a proof $\prf$ satisfies the above three conditions with respect to some node $\node$ and trace values $\traceval \in \TraceValsOf_{\Antecedents}(\node)$ and $\traceval' \in \TraceValsOf_{\Consequents}(\node)$, then there exists some $n$ such that $L_{\mathscr{B}^{(\prf, \node)}_{(\traceval, \traceval')}} \sim L_{\mathscr{A}(n)^{(\prf, \node)}_{(\traceval, \traceval')}}$ with $\mathscr{B}^{(\prf, \node)}_{(\traceval, \traceval')}$ grounded if and only if $\traceval' \sim^{\node}_{\prf} \traceval$ (for ${\sim} \in \{ {<}, {\leq} \}$). We will rely on two properties, which relate to the trace width $\TraceWidth(\prf)$ of $\prf$ (defined above in the proof of \cref{lem:ApproximateAutomata:FiniteAmbiguity}) and the (binary left-hand) cycle \emph{threshold} $\CycleThreshold(\prf)$ of $\prf$, which we define as the (necessarily finite) number of distinct triples $(\traceval, \traceval', \node)$ such that $\node \in \NodesOf(\prf)$, $\traceval \in \TraceValsOf_{\Antecedents}(\node)$ and $\traceval' \in \TraceValsOf_{\Antecedents}(\node)$. These properties are:
\begin{enumerate}[ref={(\arabic*)}]
  \item 
  \label{prop:SimpleCycles}
  Any (left- or right-hand) trace $\vec{\traceval}_{1..n}$ following a path $\vec{\node}_{1..n}$ containing strictly more than $\TraceWidth(\prf)$ occurrences of some node $\node$ \emph{must} contain a simple cycle.
  \item 
  \label{prop:SimpleBinaryCycles}
  Any pair $(\vec{\traceval}_{1..n}, \vec{\traceval'}_{1..n})$ of left-hand traces following a path $\vec{\node}_{1..n} \in \prf$ such that $n > \CycleThreshold(\prf)$ \emph{must} contain a simple binary cycle.
\end{enumerate}
 
In particular, under the conditions of \cref{def:CompletenessRestrictions}, the latter property entails that the difference between the size of two traces along some path is bounded. Let $\maxStep(\prf)$ denote the maximum value in the set $\{ \TracePairsOf^{(\node, \node')}_{1}(\traceval, \traceval') \,\mid\, (\node, \node') \in \prf \}$.

\begin{lemma}
\label{lem:LefthandTracePairs:SizeBounded}
  Let $\prf$ satisfy the conditions of \cref{def:CompletenessRestrictions} with respect to some node $\node$ and trace values $\traceval \in \TraceValsOf_{\Antecedents}(\node)$ and $\traceval' \in \TraceValsOf_{\Consequents}(\node)$, and let $\vec{\traceval}_{1..n}$ and $\vec{\traceval'}_{1..n}$ be left-hand traces rooted at $\node$ following some path $\vec{\node}_{1..n} \in \prf$; then 
  \begin{equation*}
    |\ProgPointsOf_{\vec{\node}}(\vec{\traceval}) - \ProgPointsOf_{\vec{\node}}(\vec{\traceval'})| \leq \CycleThreshold(\prf) \times \maxStep(\prf)
  \end{equation*}
\end{lemma}
\begin{proof}
  By well-founded induction on the length $n$ of the traces/path. We first note that $\ProgPointsOf_{\vec{\node}}(\vec{\traceval}) - \ProgPointsOf_{\vec{\node}}(\vec{\traceval'})$ is well-defined (and an integer) since both $\ProgPointsOf_{\vec{\node}}(\vec{\traceval})$ and $\ProgPointsOf_{\vec{\node}}(\vec{\traceval})$ are finite ordinals, which is guaranteed since $\prf$ is finitely progressing (condition (\ref{def:CompletenessRestrictions:FinitelyProg}) of \cref{def:CompletenessRestrictions}).
  \begin{description}[listparindent={\parindent},parsep=0em]
    \item[($n \leq \CycleThreshold(\prf)$):]
    The result follows since in this case $\ProgPointsOf_{\vec{\node}}(\vec{\traceval})$ and $\ProgPointsOf_{\vec{\node}}(\vec{\traceval'})$ are both bounded by $\CycleThreshold(\prf) \times \maxStep(\prf)$.
    \item[($n > \CycleThreshold(\prf)$):] 
    In this case, by the property (\ref{prop:SimpleBinaryCycles}) above, $(\vec{\traceval}, \vec{\traceval'})$ must contain a simple binary cycle. Suppose {\WLOG}that this cycle is $(\vec{\traceval}_{i..j}, \vec{\traceval'}_{i..j})$. 
    Since $(\vec{\traceval}_{i..j}, \vec{\traceval'}_{i..j})$ is a cycle, by definition we have that $\vec{\node}_{i} = \vec{\node}_{j}$ with $\vec{\traceval}_{i} = \vec{\traceval}_{j}$ and $\vec{\traceval'}_{i} = \vec{\traceval'}_{j}$. Thus $\vec{\node}_{1..i} \cdot \vec{\node}_{(j+1)..n}$ is also a path in $\prf$ which, moreover, is followed by both $\vec{\traceval}_{1..i} \cdot \vec{\traceval}_{(j+1)..n}$ and $\vec{\traceval'}_{1..i} \cdot \vec{\traceval'}_{(j+1)..n}$. All that we have done here is remove the cycle from the original path. 
    Now, notice that the length of $\vec{\node}_{1..i} \cdot \vec{\node}_{(j+1)..n}$ is strictly less than $n$, thus by the inductive hypothesis we have that
    \begin{multline*}
      | \ProgPointsOf_{(\vec{\node}_{1..i} \cdot \vec{\node}_{(j+1)..n})}(\vec{\traceval}_{1..i} \cdot \vec{\traceval}_{(j+1)..n})
      \\
        {} - \ProgPointsOf_{(\vec{\node}_{1..i} \cdot \vec{\node}_{(j+1)..n})}(\vec{\traceval'}_{1..i} \cdot \vec{\traceval'}_{(j+1)..n}) |
      \\
        {} \leq (\CycleThreshold(\prf) \times \maxStep(\prf))
    \end{multline*}
    From \cref{def:ProgressionSum}, we have the following:
    \begin{equation*}
      \ProgPointsOf_{(\vec{\node}_{1..i} \cdot \vec{\node}_{(j+1)..n})}(\vec{\traceval}_{1..i} \cdot \vec{\traceval}_{(j+1)..n}) = \ProgPointsOf_{\vec{\node}_{j..n}}(\vec{\traceval}_{j..n}) + \ProgPointsOf_{\vec{\node}_{1..i}}(\vec{\traceval}_{1..i})
    \end{equation*}
    and
    \begin{align*}
      & \ProgPointsOf_{(\vec{\node}_{1..i} \cdot \vec{\node}_{(j+1)..n})}(\vec{\traceval'}_{1..i} \cdot \vec{\traceval'}_{(j+1)..n})
        \\ & \hspace{6em} {} = \ProgPointsOf_{\vec{\node}_{j..n}}(\vec{\traceval'}_{j..n}) + \ProgPointsOf_{\vec{\node}_{1..i}}(\vec{\traceval'}_{1..i})
    \end{align*}
    Moreover,
    \begin{align*}
      & \ProgPointsOf_{(\vec{\node}_{1..i} \cdot \vec{\node}_{(i+1)..j} \cdot \vec{\node}_{(j+1)..n})}(\vec{\traceval}_{1..i} \cdot \vec{\traceval}_{(i+1)..j} \cdot \vec{\traceval}_{(j+1)..n})
        \\ & \hspace{3em} {} = \ProgPointsOf_{\vec{\node}_{j..n}}(\vec{\traceval}_{j..n}) + \ProgPointsOf_{\vec{\node}_{i..j}}(\vec{\traceval}_{i..j}) + \ProgPointsOf_{\vec{\node}_{1..i}}(\vec{\traceval}_{1..i})
    \end{align*}
    and
    \begin{align*}
      & \ProgPointsOf_{(\vec{\node}_{1..i} \cdot \vec{\node}_{(i+1)..j} \cdot \vec{\node}_{(j+1)..n})}(\vec{\traceval'}_{1..i} \cdot \vec{\traceval'}_{(i+1)..j} \cdot \vec{\traceval'}_{(j+1)..n})
        \\ & \hspace{2em} {} = \ProgPointsOf_{\vec{\node}_{j..n}}(\vec{\traceval'}_{j..n}) + \ProgPointsOf_{\vec{\node}_{i..j}}(\vec{\traceval'}_{i..j}) + \ProgPointsOf_{\vec{\node}_{1..i}}(\vec{\traceval'}_{1..i})
    \end{align*}
    Finally, since $\prf$ is balanced (condition (\ref{def:CompletenessRestrictions:Balanced}) of \cref{def:CompletenessRestrictions}), we have that $\ProgPointsOf_{\vec{\node}_{i..j}}(\vec{\traceval}_{i..j}) = \ProgPointsOf_{\vec{\node}_{i..j}}(\vec{\traceval}_{i..j})$, and so therefore
    \begin{align*}
      & \mathrlap{\ProgPointsOf_{\vec{\node}_{1..n}}(\vec{\traceval}_{1..n}) - \ProgPointsOf_{\vec{\node}_{1..n}}(\vec{\traceval}_{1..n})} \\
      & = \ProgPointsOf_{\vec{\node}_{j..n}}(\vec{\traceval}_{j..n}) + \ProgPointsOf_{\vec{\node}_{i..j}}(\vec{\traceval}_{i..j}) + \ProgPointsOf_{\vec{\node}_{1..i}}(\vec{\traceval}_{1..i}) \\
      & \hspace{2em} {} - (\ProgPointsOf_{\vec{\node}_{j..n}}(\vec{\traceval'}_{j..n}) + \ProgPointsOf_{\vec{\node}_{i..j}}(\vec{\traceval'}_{i..j}) + \ProgPointsOf_{\vec{\node}_{1..i}}(\vec{\traceval'}_{1..i})) \\
      & = \ProgPointsOf_{\vec{\node}_{j..n}}(\vec{\traceval}_{j..n}) + \ProgPointsOf_{\vec{\node}_{1..i}}(\vec{\traceval}_{1..i}) \\
      & \hspace{2em} {} - (\ProgPointsOf_{\vec{\node}_{j..n}}(\vec{\traceval'}_{j..n}) + \ProgPointsOf_{\vec{\node}_{1..i}}(\vec{\traceval'}_{1..i})) \\
      & \hspace{4em} {} + \ProgPointsOf_{\vec{\node}_{i..j}}(\vec{\traceval}_{i..j}) - \ProgPointsOf_{\vec{\node}_{i..j}}(\vec{\traceval'}_{i..j}) \\
      & = \ProgPointsOf_{\vec{\node}_{j..n}}(\vec{\traceval}_{j..n}) + \ProgPointsOf_{\vec{\node}_{1..i}}(\vec{\traceval}_{1..i}) \\
      & \hspace{4em} {} - (\ProgPointsOf_{\vec{\node}_{j..n}}(\vec{\traceval'}_{j..n}) + \ProgPointsOf_{\vec{\node}_{1..i}}(\vec{\traceval'}_{1..i})) \\
      & = \ProgPointsOf_{(\vec{\node}_{1..i} \cdot \vec{\node}_{(j+1)..n})}(\vec{\traceval}_{1..i} \cdot \vec{\traceval}_{(j+1)..n}) \\
      & \hspace{4em} {} - \ProgPointsOf_{(\vec{\node}_{1..i} \cdot \vec{\node}_{(j+1)..n})}(\vec{\traceval'}_{1..i} \cdot \vec{\traceval'}_{(j+1)..n})
    \end{align*}
    whence the result then follows.
    
    Notice that this last piece of equational reasoning above does not hold in general when the size of the traces can be infinite ordinals; however it does hold for finite ordinals (i.e.~natural numbers). 
    \qedhere
  \end{description}
\end{proof}

We now prove that, under the conditions of \cref{def:CompletenessRestrictions}, the problem of deciding trace value ordering relations is equivalent to deciding containment between finitely ambiguous weighted languages.

\begin{theorem}
\label{thm:TraceValueOrdering-Automata-Equivalence:Restricted}
  If $\prf$ satisfies the three conditions of \cref{def:CompletenessRestrictions} with respect to some node $\node$ and trace values $\traceval \in \TraceValsOf_{\Antecedents}(\node)$ and $\traceval' \in \TraceValsOf_{\Consequents}(\node)$, then 
  \begin{equation*}
    L_{\mathscr{B}^{(\prf, \node)}_{(\traceval, \traceval')}} \sim L_{\mathscr{A}(N)^{(\prf, \node)}_{(\traceval, \traceval')}} \wedge \text{$\mathscr{B}^{(\prf, \node)}_{(\traceval, \traceval')}$ grounded} \Leftrightarrow \traceval' \sim^{\node}_{\prf} \traceval
  \end{equation*}
  where ${\sim} \in \{ {<}, {\leq} \}$ and $N = 2 + \CycleThreshold(\prf) \times \maxStep(\prf) \times \TraceWidth(\prf) + \TraceWidth(\prf)$
\end{theorem}
\begin{proof}
  \begin{description}
    \item[($\Rightarrow$):]
    Immediately from \cref{cor:ApproximateAutomata:Soundness} and \cref{thm:TraceValueOrdering-Automata-Equivalence}.
    \item[($\Leftarrow$):]
    Assume that $\traceval' \sim^{\node}_{\prf} \traceval$; it follows from \cref{thm:TraceValueOrdering-Automata-Equivalence} that $L_{\mathscr{B}^{(\prf, \node)}_{(\traceval, \traceval')}} \sim L_{\mathscr{A}^{(\prf, \node)}_{(\traceval, \traceval')}}$ and that $\mathscr{B}^{(\prf, \node)}_{(\traceval, \traceval')}$ is grounded. To show that $L_{\mathscr{B}^{(\prf, \node)}_{(\traceval, \traceval')}} \sim L_{\mathscr{A}(N)^{(\prf, \node)}_{(\traceval, \traceval')}}$, take an arbitrary accepting run $\rho : q \xrightarrow{\vec{\sigma}_{1..n}} q_{n}$ of $\mathscr{B}^{(\prf, \node)}_{(\traceval, \traceval')}$. Since $L_{\mathscr{B}^{(\prf, \node)}_{(\traceval, \traceval')}} \sim L_{\mathscr{A}^{(\prf, \node)}_{(\traceval, \traceval')}}$, it follows that there is also an accepting run $\rho' : q \xrightarrow{\vec{\sigma}_{1..n}} q'_{n}$ of $\mathscr{A}^{(\prf, \node)}_{(\traceval, \traceval')}$ such that $\mathsf{V}(\rho) \sim \mathsf{V}(\rho')$. We now consider the following two exhaustive cases:
    \begin{itemize}[label={-},wide,parsep=0em,itemsep=0.5em]
      \item 
      If there is no $q'_{i+1} = \top$ ($0 < i < n$), or the number $k$ of occurrences of $\vec{\sigma}_{i+1}$ in the sequence $\vec{\sigma}_{(i+1)..n}$, where $i$ is the least such that $q'_{i+1} = \top$, is strictly less than $N$, then we have immediately by \cref{lem:ApproximateAutomata:RelativeCompleteness} that there is a run $q'' : q \xrightarrow{\vec{\sigma}_{1..n}} q''_{n}$ of $\mathscr{A}(m)^{(\prf, \node)}_{(\traceval, \traceval')}$ for all $m > 0$ or $m > k$ respectively, such that $\mathsf{V}(\rho') = \mathsf{V}(\rho'')$. Either way, this holds for for $m = N > k \geq 0$ in particular.
      \item 
      We have $k \geq N$, with $k$ the number of occurrences of $\vec{\sigma}_{i+1}$ in the sequence $\vec{\sigma}_{(i+1)..n}$ where $i$ is the least such that $q'_{i+1} = \top$. In this case, we show that there also exists a run $\rho'' : q \xrightarrow{\vec{\sigma}_{1..n}} q''_n$ of $\mathscr{A}^{(\prf, \node)}_{(\traceval, \traceval')}$ such that $\mathsf{V}(\rho') \leq \mathsf{V}(\rho'')$ and the number $k'$ of occurrences of $\vec{\sigma}_{j+1}$ in the sequence $\vec{\sigma}_{(j+1)..n}$, where $j$ is the least such that $q''_{j} = \top$, satisfies $k' < N$.
      \par
      Firstly, notice that each element of $\vec{\sigma}$ is a node in $\prf$, thus in the remainder of this case we shall write $\vec{\sigma}_{1..n}$ as $\vec{\node}_{1..n}$. Notice also that we can construct a left-hand trace $\vec{\traceval}_{1..i}$ (beginning with $\traceval$) from the states $q'_{1}, \ldots, q'_{i}$, since $q'_{m} = (\vec{\node}_{m}, \vec{\traceval}_{m})$ for each $m \leq i$. Then we have by \cref{lem:Automaton:Faithful:Antecedent} that $\vec{\traceval}_{1..i}$ follows $\vec{\node}_{1..n}$ and $\mathsf{V}(\rho') = \ProgPointsOf_{\vec{\node}}(\vec{\traceval})$. Similarly, we can construct a right-hand trace $\vec{\traceval'}_{1..n}$ (beginning with at $\traceval'$) from the states of $\rho$ which, by \cref{lem:Automaton:Faithful:Consequent}, also follows $\vec{\node}_{1..n}$ and for which $\mathsf{V}(\rho) = \ProgPointsOf_{\vec{\node}}(\vec{\traceval'})$. Notice that since $\rho$ is an accepting run, by \cref{cor:Automaton:Faithful:Consequent}, $\vec{\traceval'}$ is also a \emph{positive maximal} right-hand trace.
      \par
      Now, there are $k \geq N > \TraceWidth(\prf)$ occurrences of $\vec{\node}_{i+1}$ in the sequence $\vec{\node}_{(i+1)..n}$. Let $\vec{\node}_{m..n}$ be the (smallest) path which is a tail of $\vec{\node}_{1..n}$ containing $\TraceWidth(\prf) + 1$ occurrences of $\vec{\node}_{i+1}$. By property \ref{prop:SimpleCycles} above, $\vec{\traceval'}_{m..n}$ must contain a simple cycle, say $\vec{\traceval'}_{r..s}$ (with $m \leq r < s \leq n$). Since $\prf$ is both \emph{dynamic} and \emph{finitely progressing} we have that $0 < \ProgPointsOf_{\vec{\node}_{r..s}}(\vec{\traceval'}_{r..s}) < \omega$. This means that, via a `pumping' construction, we can build a path in $\prf$ followed by a \emph{positive maximal} right-hand trace whose size is \emph{arbitrarily} (finitely) large, as follows:
      \begin{align*}
        \text{path:} \hspace{2em} & \vec{\node}_{1..r} \cdot {\overbrace{\vphantom{I}\vec{\node}_{(r+1)..s} \cdot \ldots}^{\text{\clap{arbitrarily many occurrences of $\vec{\node}_{(r+1)..s}$}}}} \cdot \vec{\node}_{(s+1)..n}
          \\
        \text{trace:} \hspace{2em} & \vec{\traceval\mathrlap{'}}_{1..r} \cdot {\underbrace{\vec{\traceval\mathrlap{'}}_{(r+1)..s} \cdot \ldots}_{\text{\clap{matching number of occurrences of $\vec{\traceval'}_{(r+1)..s}$}}}} \cdot \vec{\traceval\mathrlap{'}}_{(s+1)..n}
      \end{align*}
      In particular, there exists such a path $\vec{\node'}$ followed by such a trace $\vec{\traceval''}$ with $\ProgPointsOf_{\vec{\node'}}(\vec{\traceval''}) \geq \maxStep(\prf) \times s$. Since $\traceval' \sim^{\node}_{\prf} \traceval$, there must exist a left-hand trace $\vec{\traceval'''}_{1..t}$ following $\vec{\node'}$ such that $\ProgPointsOf_{\vec{\node'}}(\vec{\traceval'''}) \geq \ProgPointsOf_{\vec{\node'}}(\vec{\traceval''})$. Crucially, since $\ProgPointsOf_{\vec{\node'}}(\vec{\traceval'''}) \geq \maxStep(\prf) \times s$, it must be that $t \geq s$. Consequently, $\vec{\traceval'''}_{1..s}$ follows $\vec{\node}_{1..n}$.
      \par
      To complete the case, we note the following. First, by \cref{lem:LefthandTracePairs:SizeBounded}, we have that $| \ProgPointsOf_{\vec{\node}}(\vec{\traceval}_{1..i}) - \ProgPointsOf_{\vec{\node}}(\vec{\traceval'''}_{1..i}) | \leq \CycleThreshold(\prf) \times \maxStep(\prf)$. Second, there are at least $1 + \CycleThreshold(\prf) \times \maxStep(\prf) \times \TraceWidth(\prf)$ occurrences of $\vec{\node}_{i+1}$ in the path $\vec{\node}_{(i+1)..(m-1)}$. This means, again by property \ref{prop:SimpleCycles} above, that there are (at least) $\CycleThreshold(\prf) \times \maxStep(\prf)$ basic cycles in the left-hand trace $\vec{\traceval'''}_{(i+1)..(m-1)}$. Since $\prf$ is dynamic it then follows that $\ProgPointsOf_{\vec{\node}_{(i+1)..(m-1)}}(\vec{\traceval'''}_{(i+1)..(m-1)}) \geq \CycleThreshold(\prf) \times \maxStep(\prf)$, and therefore it must be that $\ProgPointsOf_{\vec{\node}}(\vec{\traceval'''}_{1..(m-1)}) \geq \ProgPointsOf_{\vec{\node}}(\vec{\traceval}_{1..i})$. Finally, by \cref{lem:Automaton:Faithful:Antecedent}, there exists a run $\rho''' : q \xrightarrow{\vec{\node}_{1..n}} q'''_{n}$ of $\mathscr{A}^{(\prf, \node)}_{(\traceval, \traceval')}$ with $q'''_{u} = (\vec{\node}_{u}, \vec{\traceval'''}_{u})$ for each $u < m$; moreover we may assume without loss of generality that $q'''_{u} = \top$ for $m \leq u \leq n$, since the automaton may transition to the $\top$ state at any point. Hence, $\mathsf{V}(\rho''') = \ProgPointsOf_{\vec{\node}}(\vec{\traceval'''}_{1..(m-1)}) \geq \ProgPointsOf_{\vec{\node}}(\vec{\traceval}_{1..i}) = \mathsf{V}(\rho')$. Notice that $\vec{\node}_{m}$ is an occurrence of $\vec{\node}_{i+1}$, and that there must necessarily be fewer than $N$ occurrences of $\vec{\node}_{i+1}$ in the path $\vec{\node}_{m..n}$. Thus, by \cref{lem:ApproximateAutomata:RelativeCompleteness}, it follows that there is a run $\rho'' : q \xrightarrow{\vec{\node}_{1..n}} q''_{n}$ of $\mathscr{A}(N)^{(\prf, \node)}_{(\traceval, \traceval')}$ with $\mathsf{V}(\rho'') = \mathsf{V}(\rho''') \geq \mathsf{V}(\rho')$.
    \end{itemize}
    In both cases, we have a run $q'' : q \xrightarrow{\vec{\sigma}_{1..n}} q''_{n}$ of $\mathscr{A}(N)^{(\prf, \node)}_{(\traceval, \traceval')}$ such that $\mathsf{V}(\rho) \sim \mathsf{V}(\rho') \leq \mathsf{V}(\rho'')$. Notice that, by construction, this must be an accepting run, and so $\mathsf{V}(\rho'') \leq L_{\mathscr{A}(N)^{(\prf, \node)}_{(\traceval, \traceval')}}(\vec{\sigma}_{1..n})$, whence the result follows.
    \qedhere
  \end{description}
\end{proof}

Finally therefore, as a corollary, we may decide trace value relations.

\begin{corollary}
  If $\prf$ satisfies the three conditions of \cref{def:CompletenessRestrictions} with respect to some node $\node$ and trace values $\traceval \in \TraceValsOf_{\Antecedents}(\node)$ and $\traceval' \in \TraceValsOf_{\Consequents}(\node)$, then it is decidable whether or not $\traceval' \sim^{\node}_{\prf} \traceval$ holds.
\end{corollary}
\begin{proof}
  As stated previously, we may decide whether an automaton $\mathscr{B}^{(\prf, \node)}_{(\traceval, \traceval')}$ is grounded or not, and whether the conditions of \cref{def:CompletenessRestrictions} hold with respect to $\node$, $\traceval$ and $\traceval'$. Under the conditions of \cref{def:CompletenessRestrictions}, the weighted automata $\mathscr{A}(n)^{(\prf, \node)}_{(\traceval, \traceval')}$ and $\mathscr{B}^{(\prf, \node)}_{(\traceval, \traceval')}$ are finitely ambiguous sum automata, for which containment is decidable \cite{FiliotGR14}. Decidability of $\traceval' \sim^{\node}_{\prf} \traceval$ then follows from \cref{thm:TraceValueOrdering-Automata-Equivalence:Restricted}.
\end{proof}

\bibliographystyle{plain}
\bibliography{ordinal_inference}

\end{document}